\documentclass[a4paper,UKenglish]{lipics-v2018}

\title{Output-Oblivious Stochastic Chemical Reaction Networks}

\author{Ben Chugg}{The University of British Columbia, Canada}{ben.ih.chugg@gmail.com}{}{Supported by an NSERC Undergraduate Student Research Award.}

\author{Hooman Hashemi}{The University of British Columbia, Canada}{hhoomn390@gmail.com}{}{}

\author{Anne Condon}{The University of British Columbia, Canada}
{condon@cs.ubc.ca}{ https://orcid.org/0000-0003-1458-1259}{Supported by an NSERC Discovery Grant.}

\authorrunning{B. Chugg, H. Hashemi, A. Condon}

\Copyright{Ben Chugg, Hooman Hashemi, and Anne Condon}

\subjclass{Theory of computation $\to$ Models of computation $\to$ Computability, Theory of computation $\to$ Formal languages and automata theory}


\category{}

\relatedversion{}

\supplement{}

\funding{}

\EventEditors{Jiannong Cao, Faith Ellen, Luis Rodrigues, and Bernardo Ferreira}
\EventNoEds{4}
\EventLongTitle{22nd International Conference on Principles of Distributed Systems (OPODIS 2018)}
\EventShortTitle{OPODIS 2018}
\EventAcronym{OPODIS}
\EventYear{2018}
\EventDate{December 17–19, 2018}
\EventLocation{Hong Kong, China}
\EventLogo{}
\SeriesVolume{122}
\ArticleNo{0} 
\nolinenumbers 

\usepackage{amsmath,amssymb,amsthm}
\usepackage{physics}
\graphicspath{{images/}}
\usepackage{xcolor}

\DeclareMathOperator{\dom}{\textnormal{Dom}}

\newtheorem{claim}{Claim}

\renewcommand{\S}{\mathcal{S}}
\newcommand{\y}{\vb{y}}
\newcommand{\on}{|}
\newcommand{\T}{\mathcal{T}}
\newcommand{\Z}{\mathbb{Z}}
\newcommand{\N}{\mathbb{N}}
\newcommand{\Grid}{\mathcal{G}}
\newcommand{\G}{\Grid}
\newcommand{\offset}{{\bf o}}
\newcommand{\p}{p} 

\newcommand{\Q}{\mathbb{Q}}
\newcommand{\x}{\vb{x}}
\newcommand{\n}{{\bf n}}
\newcommand{\nj}[1]{{\bf n^{(#1)}}}
\newcommand{\vp}{\varphi}

\renewcommand{\a}{\vb{a}}
\newcommand{\la}{\langle}
\newcommand{\ra}{\rangle}

\renewcommand{\v}{\vb{v}}
\renewcommand{\u}{\vb{u}}
\renewcommand{\c}{\vb{c}}
\newcommand{\lx}{l_{\bf x}}
\newcommand{\lz}{l_{\bf z}}

\newcommand{\bx}{{\bf x}}
\newcommand{\z}{{\bf z}}

\newcommand{\C}{\mathcal{C}} 
\newcommand{\ZZ}{{Z}}
\newcommand{\spec}{\mathcal{Z}}

\newcommand{\inp}{\mathcal{I}}
\newcommand{\outp}{\mathcal{O}}
\newcommand{\exec}{\mathcal{E}}

\newcommand{\reac}{\mathcal{R}}

\begin{document}

\maketitle

\begin{abstract}
We classify the functions $f:\N^2 \rightarrow \N$ which are stably computable by \emph{output-oblivious} Stochastic Chemical Reaction Networks (CRNs), i.e., systems of reactions in which output species are never reactants. While it is known that precisely the semilinear functions are stably computable by CRNs, such CRNs sometimes rely on initially producing too many output species, and then consuming the excess in order to reach a correct stable state. These CRNs may be difficult to integrate into larger systems: if the output of a CRN $\C$ becomes the input to a downstream CRN $\C'$, then $\C'$ could inadvertently consume too many outputs before $\C$ stabilizes. If, on the other hand, $\C$ is output-oblivious then $\C'$ may consume $\C$'s output as soon as it is available. In this work we 
prove that a semilinear function $f:\N^2 \rightarrow \N$ is stably computable by an output-oblivious CRN with a leader if and only if it is both increasing and either \emph{grid-affine} (intuitively, its domains are congruence classes), or the minimum of a finite set of \emph{fissure functions} (intuitively, functions behaving like the min function).
 \end{abstract}

\section{Introduction}
Stochastic Chemical Reaction Networks (CRNs)---systems of reactions involving chemical species---have traditionally been used to reason about extant physical systems, but are currently also of strong interest as a distributed computing model for describing molecular programs~\cite{chen2014deterministic,soloveichik}. They are closely related to Population Protocols \cite{alistarh2015fast,angluin2006,angluin2007computational,delporte2007secretive}, another distributed computing model; these models have found applications in areas as diverse as signal processing \cite{jiang2012digital}, graphical models \cite{napp2013message}, neural networks~\cite{hjelmfelt1991chemical}, and modeling cellular processes~\cite{arkin1998stochastic,cardelli2012cell}. 
CRNs can simulate 
Universal Turing Machines \cite{soloveichik,AAE06}. However, these simulations have drawbacks: the number of reactions or molecules may scale with the space usage and the computation is only correct with an arbitrarily small probability of error. If we require {\em stable computation}---that the CRN always eventually produces the correct answer---then Angluin et al.~\cite{angluin2007computational} showed that precisely the class of semilinear predicates can be stably computed. Chen et al.~\cite{chen2014deterministic} extended this result to show that precisely the semilinear functions can be stably computed. 

Recent advances in physical implementations of CRNs and, more generally, chemical computation using strand displacement systems (e.g., \cite{qian2011scaling,soloveichik2010dna,thachuk2012space,zhang2013integrating}) are a step towards the use of CRNs in biological environments and nanotechnology. As these systems become more complex, it may be necessary to integrate multiple, interacting CRNs in one system. However, current CRN constructions may perform poorly in such scenarios. 
As a concrete example, consider a CRN $\C$ given by the reactions $X\to 2Y$, $Y+L\to \emptyset$, where the system begins with $n$ copies of input species $X$, and one copy of $L$ (called the \emph{leader}). This CRN eventually produces $2n-1$ copies of output species $Y$, and so (stably) computes the function $n\mapsto 2n-1$. If another CRN $\C'$ uses the output of $\C$ as its input, and if the first reaction occurs $n$ times before the second occurs at all, then $\C'$ may consume all $2n$ copies of $Y$ and may thus itself produce an erroneous output.
Current CRN constructions circumvent this issue by using {\em diff-representation}, where the count $y$ of output species $Y$ of a CRN is represented indirectly as the difference $y=y^P-y^C$ between the counts of two species $Y^P$ and $Y^C$~\cite{chen2014deterministic}, rather than as the count of one output species $Y$. 
While these constructions enable the counts of both $Y^C$ and $Y^P$ to be non-decreasing throughout the computation, it is not immediately clear how a second CRN might use these two species reliably as input. 

More generally, if multiple function-computing CRNs comprise a larger system it can be desirable that no CRN ever produces a number of outputs that exceeds its function value. We might even demand more: that an output species of a CRN is never used as a reactant species, i.e., is never consumed. This ensures that any secondary CRN relying on the first's output can consume the output indiscriminately.

It is thus natural to ask: What functions can be stably computed in an {\em output-oblivious manner}, in which outputs are never reactants, without using diff-representation? 
 
This question is the focus of this paper. Doty and Hajiaghayi~\cite{doty-hajiaghayi-15} already observed that output-oblivious functions must not only be semilinear but also increasing, that is, $f(\n_1) \le f(\n_2)$ whenever $\n_1 \le \n_2$, but did not provide further insights. Chalk et al. \cite{Chalk-etal-2018} asked the same question but for a different model, namely mass-action CRNs. That model tracks real-valued species concentrations, unlike
the stochastic model in which configurations are vectors of species counts.
In contrast with the mass-action mode, leader molecules can play a very important role in the stochastic model, and we focus on the case where leaders are present. Mass-action CRN models cannot have leaders since there are no species counts.
Functions that are stably computable by output-oblivious mass-action CRNs must be super-additive \cite{Chalk-etal-2018}, that is $f(n) + f(n') \le f(n+n')$. Semilinear functions that are super-additive are a proper subset of the class of output-oblivious functions (characterized in this paper) that can be stably computed by stochastic CRNs with leaders.

\subsection{Our Results}
In this work we characterize the class of output-oblivious semilinear functions, i.e., those functions that can be stably computed by an output-oblivious stochastic CRN. We assume that one copy of a leader species is present initially in addition to the input. We focus on functions with two inputs and one output, since this case already is quite complex. Our results generalize trivially when there are more outputs since each output can be handled independently, and we believe that our techniques also generalize to multiple inputs.
 
 Our results also hold for Population Protocols, since stable function-computing CRNs can be translated into Population Protocols and vice versa.
 Section \ref{sec:prelims} introduces the relevant background in order to formally describe our results, but we describe them informally here.
 
Perhaps the simplest type of output-oblivious function with domain $\N^2$ is an affine function, such as $f(n_1,n_2) = 2n_1 + 3n_2 +1$ which could be computed by a CRN with reactions $L\to Y$, $X_1\to 2Y$ and $X_2\to 3Y$ where $L$ is a single leader. Here and hereafter, $X_i$ will typically correspond to the input species representing $n_i$.

In Section \ref{sec:sufficiency_proof} we show that an increasing function that can be specified as partial affine functions whose domains are different ``grids'' of $\N^2$ is also output-oblivious; for example, the function $f(n_1,n_2) = 2n_1 + 3n_2 +1$ when $n_1+n_2 = 0 \mod 2$, and  $f(n_1,n_2) = 2n_1 + 3n_2$ when $n_1+n_2 = 1 \mod 2$. 
More generally, a function that can be specified in terms of output-oblivious partial functions $f_i$, $1\le i \le k$, defined on different grids of $\N^2$, is output-oblivious. The grids may be 0-dimensional, in which case they are points; 1-dimensional in which case they are lines, or 2-dimensional.  We call such functions {\em grid-affine} functions. See Figure \ref{fig:grid_affine} for a slightly more complicated example of a grid-affine function, and a representation of its domains. We show how the CRNs for partial functions $f_i$ on the different grids can be ``stitched'' together to obtain an output-oblivious CRN for $f$.

\begin{figure}
\centering
\begin{subfigure}[]{0.44\textwidth}
\centering
\includegraphics[scale=0.4]{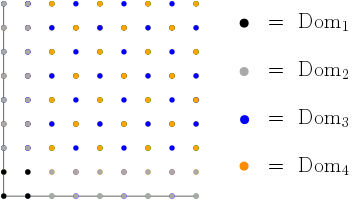}
\end{subfigure}
\begin{subfigure}[]{0.54\textwidth}
\centering
\[
f(\n) = \begin{cases} 
\vp_1(\n) = 2, & \n\in\dom_1,\\
\vp_2(\n) = n_1+n_2,& \n\in\dom_2,\\
\vp_3(\n) = n_1+2n_2+5, &\n\in\dom_3,\\
\vp_4(\n) = n_1+2n_2+4, &\n\in\dom_4.
\end{cases}
\]
\end{subfigure}
\caption{
Here we represent a grid-affine function $f:\N^2\rightarrow \N$ by its decomposition on different domains, all of which are grids. The domains of $f$ are illustrated on the left. Each black point is a zero-dimensional grid, while the grey points represent four one-dimensional grids, namely the lines $\{(\alpha,0)+(2,i):\alpha\in\N\}$ and $\{(0,\alpha)+(i,2):\alpha\in\N\}$ for $i=0,1$. The blue points represent points $(n_1,n_2)$ such that $n_1+n_2$ is even, and cover the union of two grids: $\{(2\alpha_1,2\alpha_2):\alpha_i\in\N\}\cup\{(2\alpha_1,2\alpha_2)+(1,1):\alpha_i\in\N\}$. Similarly, the gold points represent two grids.}
\label{fig:grid_affine}
\end{figure}

It is also straightforward to obtain an output-oblivious CRN for a function $f$ that is the min of a finite set of output-oblivious functions. In the simplest case, for example, $\min(n_1,n_2)$ can be computed as $X_1+X_2\to Y$. In our main positive result we describe a more general type of ``min-like'' function, which we call a {\em fissure function}, and we show how to construct output-oblivious CRNs for such functions. We give a very simple example of a fissure function and a corresponding output-oblivious CRN in Figure \ref{fig:simplefissure}.

However, constructing CRNs for other fissure functions appears to be significantly trickier than that shown in Figure \ref{fig:simplefissure}. Consider the function
$f(n_1,n_2) = 2n_1+3n_2+2$ if $n_1 > n_2$, $f(n_1,n_2) = 3n_1+2n_2+2$ if $n_1 < n_2$ and $f(n_1,n_2) = 5n_1$ on the "fissure line" $n_1 = n_2$. The simple line-tracking mechanism of the CRN of Figure \ref{fig:simplefissure} can't be used here because the affine functions for the "wedge" domains "$n_1>n_2$" and "$n_1<n_2$" depend both on $n_1$ and $n_2$. Also the function cannot be written as the sum of an increasing grid-affine function and an increasing simple fissure function of the type in Figure \ref{fig:simplefissure}, where the "above" function $\vp_A()$ depends only on $n_1$ and the "below" function $\vp_B()$ depends only on $n_2$. Our main positive result is a construction that can handle such fissure functions, as well as functions with multiple parallel fissure lines.

\begin{figure}
\begin{subfigure}[]{0.265\textwidth}
\includegraphics[scale=0.3]{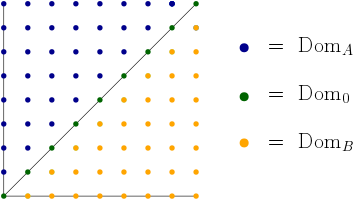}
\end{subfigure}
\begin{subfigure}[]{0.44\textwidth}
\[
f(\n) = \begin{cases}
\vp_A(\n)= n_1+1,    &\n\in\dom_A, \\
\vp_0(\n) = n_1,      & \n\in\dom_0, \\
\vp_B(\n) = n_2+1,    & \n\in\dom_B.
\end{cases}
\]
\end{subfigure}
\begin{subfigure}[]{0.2\textwidth}
\centering
\[
\begin{array}{lll}
{[0]} + X_1 &\rightarrow  {[A]} + Y \\
{[0]} + X_2 &\rightarrow {[B]} + Y \\
{[A]} + X_2 &\rightarrow {[0]} \\
{[B]} + X_1 &\rightarrow {[0]}
\end{array}
\]
\end{subfigure}
\caption{
A simple fissure function $f:\N^2\rightarrow \N$. On the left the three domains of $f$ are illustrated. There is one "fissure line" called $\dom_0$, and two "wedge" domains called $\dom_A$ and $\dom_B$ ("A" is above and "B" is below the fissure line). The function value on each of these domains is specified in the center. 
The function $f$ agrees with the function $\min\{n_1+1, n_2+1\}$
except that it dips down by 1 on the fissure line $\dom_0$. On the right is a CRN which stably computes $f$. 
In the CRN, the input $\n = (n_1,n_2)$ is represented as counts of species $X_1$ and $X_2$ and the leader is initially $[0]$. The three possible states ${[0], [A]}$ and ${[B]}$ of the leader track
whether the input lies on the fissure line $\dom_0$, which is the line where 
$\vp_A(\n)-\vp_B(\n)=0$, or whether the input lies above or below the fissure line, i.e., in domains $\dom_A$ or $\dom_B$ respectively. In this simple example, the CRN need not track how far above (or how far below) the fissure line an input might be, since the function $\vp_A$ does not depend on $n_2$ (and the function $\vp_B$ does not depend on $n_1$).}
\label{fig:simplefissure}
\end{figure}

In Section \ref{sec:necessity_proof} we present results on the negative side. A non-trivial example of a function that is not output-oblivious is the  maximum function. Intuitively, a CRN that attempts to compute the max would have to keep track of the relative difference of its two inputs in order to know when the count of one input overtakes the count of the other, and it's not possible to keep track of that difference with a finite number of states. Developing this intuition further, we show that an increasing semilinear function $f:\N^2 \rightarrow \N$ is output-oblivious if and only if $f$ is grid-affine or is the min of finitely many fissure functions. 

Putting both positive and negative results together, we state our main result here (see Section \ref{sec:prelims} for precise definitions of grid-affine and fissure functions). 

\vspace{.1in}

\begin{theorem}
\label{thm:output_oblivous_iff}
A semilinear function $f:\N^2\to\N$ is output-oblivious if and only if $f$ is increasing and is either grid-affine or the minimum of finitely many fissure functions. 
\end{theorem}

Since only semilinear functions are stably computable by CRNs, Theorem \ref{thm:output_oblivous_iff} provides a complete characterization of functions $f:\N^2\to\N$ which are output-oblivious. Moreover, in Section \ref{sec:necessity_proof}, we will prove that if $f$ is output-\emph{monotonic}, then it is either grid-affine or the minimum of fissure functions, a stronger statement than in Theorem \ref{thm:output_oblivous_iff}. 
A function is output-monotonic if it is stably computable by a CRN whose output count never decreases but unlike an output-oblivious CRN an output may act as a catalyst of a reaction, being both a reactant and product. 
For example, the CRN $X\to Y$, $L+Y\to 2Y$ which computes the function $n\mapsto n+1$ for $n \ge 1$ and $0 \mapsto 0$ is output-monotonic, but not output-oblivious. 
Thus, we also obtain a characterization for output-monotonic functions. 

To obtain our results, we provide new characterizations of semilinear sets and functions. We show that all semilinear sets can be written as finite unions of sets which are the intersection of grids and hyperplanes. Such sets are points, lines or wedges (pie-shaped slices) on 2D grids. Using this and the representation of semilinear functions as piecewise affine functions discovered by Chen et al.~\cite{chen2014deterministic}, we give a new representation of semilinear functions as ``periodic semiaffine functions'', essentially piecewise affine functions whose domains are points, lines or wedges. 

The rest of the paper is structured as follows. Section \ref{sec:prelims} provides the relevant technical background on CRNs, stable computation and semilinear functions. It also contains our new results on the structure of semilinear sets and functions, and rigorous definitions of grid-affine and fissure functions. 
In the remaining two sections we prove Theorem \ref{thm:output_oblivous_iff}, with Section \ref{sec:sufficiency_proof} providing explicit constructions of CRNs and Section \ref{sec:necessity_proof} proving that any function which is stably computable by an output-oblivious CRN obeys certain properties.

\section{Preliminaries}
\label{sec:prelims}
We begin by introducing Chemical Reaction Networks, and what it means for a CRN to stably compute a function. We then formally define grid-affine and fissure functions and, along the way, state new results concerning semilinear sets and functions. 

\subsection{Chemical Reaction Networks (CRNs)}
\label{sec:prelims_CRNS}
CRNs specify possible behaviours of systems of interacting \emph{species}. Let $\spec=\{\ZZ_1,\ldots,\ZZ_m\}$ be a finite set of species. 
At any given instant, the system is described by a configuration 
$\c\in\N^\spec$,
where $c(\ZZ_i)$ is the current count of the species $\ZZ_i\in\spec$ in the system.
The system's configuration changes by way of \emph{reactions}, each of which is described as a pair
$(\vb{s},\vb{t})=((s_1,\ldots,s_m),(t_1,\dots,t_m)) \in \N^\spec \times \N^\spec$ such that for at least one
$1\leq j\leq m$, $s_j\neq t_j$. 
Reaction $(\vb{s},\vb{t})$ can be written as
 \[\sum_{k:s_k> 0} s_k\ZZ_k\rightarrow \sum_{k:t_k>0} t_k\ZZ_k.\] 
The species $\ZZ_k$ with $s_k>0$ are the \emph{reactants}, which are {\em consumed}, while those with $t_k>0$ are the \emph{products} (if both $s_k>0$ and $t_k>0$ then species $Z_k$ is a catalyst). A CRN is thus formally described as a pair $\C=(\spec,\reac)$, where $\spec$ is a set of species, and $\reac$ a set of reactions.  
Reaction $r=(\vb{s},\vb{t})$ is \emph{applicable} to configuration $\c$ if $\vb{s}\leq \c$ (pointwise inequality), i.e., sufficiently many copies of each reactant are present. If applicable reaction $(\vb{s},\vb{t})$ occurs when the system is in configuration $\c=(c_1,\dots,c_m)$,
the new configuration is $\c'=(c_1-s_1+t_1,\dots,c_m-s_m+t_m)$. In this case we say 
that $\c'$ is \emph{directly reachable} from $\c$ and
write $\c\xrightarrow{r}\c'$. An \emph{execution} $\exec=\c_0,\dots,\c_t$ of $\C$ is a sequence of configurations of $\C$ such that $\c_i$ is directly reachable from $\c_{i-1}$ for $1 \le i \le t$. We say that $\c_t$ is \emph{reachable} from $\c_0$.  \\

\noindent
\textbf{Stable CRN Computation of functions with a leader.} 
Angluin et al. \cite{angluin2006} introduced the concept of stable computation of boolean predicates by population protocols,
and Chen et al. \cite{chen2014deterministic} adapted the notion to function computation by CRNs. While this paper focuses on two-dimensional domains, we present the following details in full generality.

Let $f:\N^k\to\N^\ell$ be a function. Formally, a \emph{Chemical Reaction Network (CRN) for computing $f$ with a leader} is $\C=(\spec,\reac,\inp,\outp, L)$, where $\spec$ is a set of  species, $\reac$ is a set of reactions, $\inp =\{X_1,X_2,\ldots,X_k\} \subseteq \spec$ is an ordered set of input species, $\outp=\{Y_1,Y_2,\ldots,Y_{\ell}\} \subseteq \spec$ is an ordered set of output species and $L$ is a leader species, $L \in \spec \setminus \inp$.

Function computation on input $\n = (n_1,\ldots,n_k) \in \N^k$ starts from a \emph{valid initial configuration} $\c_0$ of $\C$; namely a configuration in which the count of $L$ is 1, 
the count of species $X_i$ is $n_i$, and the count of any other species is 0. 
A \emph{computation} is an execution of $\C$ from a valid initial configuration to a stable configuration. A configuration $c$ is \emph{stable} if 
for every $\c'\in\N^m$ reachable from $\c$, $\c(Y)=\c'(Y)$ for all $Y\in \outp$. That is, once the system reaches configuration $\c$, the counts of the output species do not change. We say that $\C$ \emph{stably computes $f$} if for every valid initial configuration $\c_0$ and for every configuration $\c$ reachable from $\c_0$, there exists a
stable configuration $\c'$ reachable from $\c$ such that
$f(\c_0(X_1),\dots,\c_0(X_k))=(\c'(Y_1),\dots,\c'(Y_\ell))$. \\

\noindent
\textbf{Output-monotonic and output-oblivious CRNs.}
We say a CRN $\C$ is \emph{output-oblivious} if it never consumes any of its output species, and \emph{output-monotonic} if on all executions from a valid initial configuration, the count of any output species never decreases. As noted in the introduction, these notions are not equivalent. We say a function $f$ is \emph{output-oblivious} (\emph{monotonic}) if there exists an output-oblivious (monotonic) CRN which stably computes $f$. Our results show that the set of output-oblivious functions and output-monotonic functions are the same.

\subsection{Linear and Semilinear Sets; Lines, Grids, and Wedges}
\label{sec:prelims_sets}
For a vector $\v$, let $v_i$ denote its $i$th coordinate. 
Let $D\subseteq\N^2$ and let $\Pi_1$ and $\Pi_2$ denote the projection maps onto $x$ and $y$ axes, respectively. We say $D$ is \emph{two-way-infinite} if $|\Pi_1(D)|=|\Pi_2(D)|=\infty$, \emph{one-way-infinite} if either $|\Pi_1(D)|=\infty$ or $|\Pi_2(D)|=\infty$ but not both, and finite if $|\Pi_1(D)| < \infty$ and $|\Pi_2(D)|<\infty$. Also, if $A,B \subseteq \N^2$ and $\n\in\N^2$ we let $A+B=\{a+b:a\in A,b\in B\}$ and $A+\n=A+\{\n\}$. 

A set $E\subset\N^2$ is {\em linear} if $E=\{\sum_{i=1}^t \x_i\alpha_i+\offset:\alpha_i\in\N$\} for some $t\in\N$ and $\x_i,\offset\in\N^2$. 
If $t=1$ we say that $E$ is a \emph{line}. 
A set is \emph{semilinear} if it is the finite union of linear sets. 

A linear set $\G \subseteq \N^2$ is a \emph{grid} if there exist $p,q\in\N$ and $\offset\in\N^2$ such that $\G=\{(p,0)\alpha_1+(0,q)\alpha_2:\alpha_i\in\N\} +\offset =\{(p\alpha_1+o_1,q\alpha_2+o_2):\alpha_i\in\N\}$. 
If both $p$ and $q$ are zero, the grid is simply the point $\offset$. If $p>0$ and $q=0$, or $p=0$ and $q>0$, the grid is a one-way-infinite line with period $p$ or $q$ respectively. If $p=q>0$ we say that the grid is periodic, with period $p$. We let $\Grid_\p +\offset$ be the grid $\{(\alpha_1p,\alpha_2p):\alpha_i\in\N\}+\offset$ and write $\Grid_p$ if $\offset=(0,0)$. 

A \emph{threshold set} is a semilinear set with the form $\{\n:\n\cdot\v\geq r\}$ (i.e., a halfspace) for some $\v \in \Z^2$ and $r\in\Z$ \cite{angluin2007computational}. Let $E$ be a two-way-infinite linear set of the form $\G\cap \T$, where $\G$ is a grid and $\T$ is a finite intersection of threshold sets. $E$ is bounded by two lines (represented by threshold sets and/or lines parallel to the x or y axes; the points on these lines, if any, are in $E$). 
(Note that the boundary of a threshold set $\{\n:(n_1,n_2)\cdot(v_1,v_2)\geq r\}$ can be written as the linear set $\{(v_2,-v_1)\alpha+(|\min\{0,r\}|,\max\{0,r\}):\alpha\in\N\}$. For example, $\la n_1,n_2\ra \cdot\la -3,2\ra = -1$ is the set $(2,3)\alpha + (1,0)$. If a line is infinite then $(v_2,-v_1)\geq (0,0)$.)
If the two bounding lines are parallel, $E$ is the finite union of lines on $\Grid$, i.e., all points of each line lie on grid $\Grid$.
Otherwise 
  we call $E$ a \emph{wedge} on $\Grid$. For example, the sets $\{\n:n_1\geq n_2\}$ and $\{(1,1)\alpha_1+(1,2)\alpha_2:\alpha_i\in\N\}$ are wedges on $\Grid_1$. Likewise, the two regions above and below the fissure line in Figure~\ref{fig:simplefissure} are wedges on $\Grid_1$.
More generally, we can intuitively think of a wedge as a pie-like slice of $\N^2 \cap \Grid$, except that pieces may be chopped off near the narrow "corner" that is closest to the origin. If the two bounding lines are the x and y axes, the wedge is all of $\Grid$. We can show the following characterization of semilinear sets; see Appendix \ref{app:set_characterization} for the proof and a more formal definition of a wedge.

\begin{lemma}
\label{lem:semi_affine_characterization}
Every semilinear 
set can be represented as the finite union of points, lines on grids, and wedges on grids, 
with all grids having the same period. 
\end{lemma}

\subsection{Semilinear, Semiaffine, Grid-Affine, and Fissure Functions}
\label{sec:prelims_functions}
For a function $f:\N^2\to\N$,
the \emph{restriction of $f$ to domain $D\subseteq \N^2$} is the partial function $f\on_D:D\to\N$ given by $f\on_D(\n)=f(\n)$ for all $\n\in D$. We say that $f:D \rightarrow \N$ is {\em (partial) affine} if $f(\n) = a_1 n_1 + a_2 n_2 + a_0$ for rational numbers $a_0,a_1$, and $a_2\in \Q$. Function $f$ is a \emph{finite combination} of the finite set of functions $\{\vp_1,\dots,\vp_k\}$ if 
$\dom(f)=\bigcup_{i=1}^k \dom(\vp_i)$ and $f(\n)=\vp_i(\n)$ whenever $\n\in\dom(\vp_i)$. Throughout we write $\dom_i$ in place of $\dom(\vp_i)$.
We define semilinear functions using a characterization of Chen et al. \cite{chen2014deterministic}: 
\begin{definition}[Semilinear function \cite{chen2014deterministic}]
\label{def:semilinear-classification}
A function $f:\N^2\to\N$ is
\emph{semilinear} if and only if $f$ is a finite combination
of partial affine functions 
with linear domains.
\end{definition}

We next define \emph{semiaffine} functions, a refinement of Definition \ref{def:semilinear-classification}. Lemma \ref{lem:semi_affine_function} then states that semilinear and semiaffine functions are equivalent. The proof is in Appendix \ref{app:semi_affine_functions}.

\begin{definition}[Semiaffine function]
\label{def:semi_affine_function}
Let $\Grid_p + \offset$ be a periodic grid. A function $f:\G_p + \offset \to\N$ is \emph{semiaffine} if and only if $f$ is a finite combination of partial affine functions whose domains are points, lines or wedges on grid $\G_p + \offset$.
A function $f:\N^2\to\N$ is \emph{semiaffine with period} $p \in \N^+$ if and only if $f$ is a combination of semiaffine functions on grids of the form $\Grid_p+\offset$.
\end{definition}

\begin{lemma}
\label{lem:semi_affine_function}
A function $f:\N^2\to\N$ is semilinear if and only if $f$ is semiaffine. 
\end{lemma}

Our main result, Theorem \ref{thm:output_oblivous_iff}, shows that output-oblivious functions are exactly the following two special types of semiaffine functions.
In the first special case, on each grid $\Grid_p + \offset$, $f$ is
restricted to be an affine (rather than a more general semiaffine)
function.

\begin{definition}[Grid-affine function]
\label{def:grid-affine}
A function $f:\N^2 \rightarrow \N$ is \emph{grid-affine} if and only if for some $p\in \N^+$, $f$ is a combination of affine functions on points and on grids of period $p$. 
\end{definition}

A function $f:D \rightarrow \N$ is \emph{increasing} if $f(\n)\leq f(\n')$ for all $\n\leq \n'$, where $\n, \n' \in D$. Doty and Hajiaghayi~\cite{doty-hajiaghayi-15} observed that an output-oblivious function must be increasing. We prove this formally in Appendix~\ref{app:non_increasing_functions}. Accordingly, we hereafter focus on increasing functions. 

\begin{definition}[Fissure function]
\label{def:fissure_function}
Let $\G$ be a two-way-infinite grid. An increasing semiaffine function $f:\Grid \rightarrow \N$ is a {\em partial fissure function} if for some $\offset\in \N^2$, $f$ can be represented as follows for all $\n \ge \offset$:
\begin{equation}
\label{eqn:fissure}
f(\n)=\begin{cases}
\vp_A(\n),&\text{if } \vp_A(\n) - \vp_B(\n) \le -k, \\
\vp_{-i}(\n) = \vp_A(\n) - d_{-i},&\text{if } \vp_A(\n) - \vp_B(\n) = -i, 1\le i < k, \\
\vp_i(\n) = \vp_B(\n) - d_{i},&\text{if }  \vp_A(\n) - \vp_B(\n) = i, 0\le i < k, \\
\vp_B(\n),&\text{if } \vp_A(\n) - \vp_B(\n) \ge k. \\
\end{cases}
\end{equation}
where $\vp_A(\n) = A_0 + A_1n_1 + A_2 n_2$, $\vp_B(\n) = B_0 + B_1n_1 + B_2 n_2$, 
for integers $A_0$ and $B_0$, nonnegative rationals $A_1,A_2, B_1$ and $B_2$, and nonnegative integers $d_{-k}, \ldots, d_{-1}, d_0, d_1, \ldots,d_k$.
For $-k \le i \le k$,
we refer to the line
$\vp_A(\n) - \vp_B(\n) = i$
as a {\em fissure line} and call it $L_i$.
Moreover, 
$\vp_A < \vp_B$ on $\dom_A$ and $\vp_B < \vp_A$ on $\dom_B$; thus
$A_1> B_1$ and $B_2 > A_2$. We say  
$f:\N^2\to\N$ is a \emph{(complete) fissure function} if 
$f$ is a  combination of partial fissure functions on grids of period $p$.
\end{definition}

\section{Proof of Sufficiency in Theorem \ref{thm:output_oblivous_iff}}
\label{sec:sufficiency_proof}

This section shows that if an increasing semilinear function $f:\N^2 \rightarrow \N$ is either a grid-affine function or a fissure function, then $f$ is output-oblivious. We do this in three lemmas. Lemma \ref{lem:partialgridoo} shows that an increasing affine function whose domain is a grid is output-oblivious.
Lemma \ref{lem:fisureoo} shows that a partial fissure function is output-oblivious.
Finally, Lemma \ref{lem:stiching} shows that if $f$ is increasing and is a combination of partial output-oblivious functions defined on grids, we can stitch together the CRNs for the partial functions to obtain an output-oblivious CRN for $f$.

\begin{lemma}
\label{lem:partialgridoo}
Let $\Grid$ be a grid. Any increasing affine function $f: \Grid \rightarrow \N$ is output-oblivious.
\end{lemma}

\begin{proof}
We consider the case that $\Grid=\{(p,0)\alpha_1+(0,q)\alpha_2:\alpha_i\in\N\}+\offset$ is two-way-infinite;
the cases when $\Grid$ is a point or a line are simpler.
Let $f(\n) = a_1 n_1 + a_2 n_2 + a_0$, where $a_1,a_2 \in\Q^+$ and $a_0 \in \Q$.
Since $\Grid$ is two-way-infinite and $f$ is increasing, $a_1$ and $a_2$ are nonnegative.
On input $\n = (n_1, n_2) \in \Grid$, i.e., given $n_1$ copies of $X_1$ and $n_2$ copies of $X_2$, the following CRN will produce $f(\n)$ copies of $Y$:
\[
\begin{array}{lll}
L + o_1 X_1 + o_2 X_2 & \rightarrow L' + (a_1o_1 + a_2o_2 +a_0)Y & \mbox{// base case} \\
L' + p X_1 & \rightarrow L' + a_1 p Y & \\
L' + q X_2 & \rightarrow L' + a_2 q Y &
\end{array}
\]
Note that the first reaction must produce a	non-negative and integral number of $Y$'s since $f(\offset)\in\N$. Likewise, $a_1p\in\N$ since $a_1p=f(\offset+(p,0))-f(\offset)$, and similarly for $a_2q$. 
Finally, the CRN is clearly output-oblivious since the output species $Y$ is never a reactant.
\end{proof}

We show in Lemma \ref{lem:fisureoo} below that any partial fissure
function is output-oblivious. First we describe some useful structure pertaining to partial fissure functions $f: \Grid \rightarrow \N$. 
We can represent such a fissure function as $f(\n) =
\min\{\vp_A(\n),\vp_B(\n)\} - d_i$, where $d_i$ is determined by the
fissure line $L_i$ on which $\n$ resides, and $d_i=0$ if $i$ is not on
a fissure line; this formulation is not identical to but is equivalent to that of Definition \ref{def:fissure_function}.
As noted in that definition, it must be that $A_1>B_1$ and $B_2>A_2$, since $\vp_A < \vp_B$ on $\dom_A$ and vice versa.

For all integers $i$, let $L_i$ be the line $\vp_A(\n) -
\vp_B(\n) = i$.  All of these lines, which include the $2k-1$
``fissure lines'' $L_i, -k < i < k$, have the same slope. In addition to the fissure lines, our CRN construction will also refer to the lines $L_{i}$ for $i$ in the range $[k,\ldots, K - 1]$, where
$K=k+d_{\max}-1$. We call these the {\em lower boundary lines}, and we call the lines $L_i$ for $i$ in the range $[-K+1,\ldots,-k]$ the {\em upper boundary lines}. Note that $(0,0)$ is on the line $L_{A_0-B_0}$ and more generally, if point ${\bf
  p}$ is on line $L_i$ then $(A_1-B_1)p_1 - (B_2-A_2)p_2 = i-A_0+B_0$.
  For $\n \in \Grid$ let
$M(\n) = (\vp_A(\n),\vp_B(\n))$.
The next lemma shows that
$M(\n) \in \N^2$ for all sufficiently large $\n \in \Grid$, even though $\dom_A$ and $\dom_B$ are proper subsets of $\Grid$.
\begin{lemma}
\label{lem:well_defined_outside_wedge}
Let $\vp:D\rightarrow \N$ be a partial affine function, where $D$ is a wedge domain on $\G$. Let $\vb{m}$ be a minimal point of $D$. Then $\vp(\n) \in\N$ 
on all $\n\in \G$ with $\n\geq \vb{m}$.
\end{lemma}
We let
${\mathcal P}$ be the set of rational points ${\bf p}$ for which $M({\bf p}) \in \N$ and let ${\mathcal Q}$ be the range of $M$ with respect to domain
${\mathcal P}$. For ${\bf q} \in {\mathcal Q}$, let $M^{-1}({\bf q})$
denote the inverse of $M$ ($M^{-1}{\bf q}$ is unique since $(A_1,A_2)$ and $(B_1,B_2)$ are linearly independent). 
The following claim follows easily from the definition of $M$ and will be useful later. \\
\begin{claim}
\label{claim:M-increasing}
Let $z_1, z_1', z_2, z_2' \in \N$.
If $z_1 \le z_1'$ and $M^{-1}(z_1,z_2)$ is in $\dom_B$
then $M^{-1}(z_1',z_2)$ is also in $\dom_B$.
Similarly if $z_2 \le z_2'$ and $M^{-1}(z_1,z_2)$ is in $\dom_A$
then $M^{-1}(z_1,z_2')$ is also in $\dom_A$. \\
\end{claim}

\begin{lemma}
\label{lem:fisureoo}
	Any partial fissure function $f: \Grid \rightarrow \mathbb{N} $ is output-oblivious.
\end{lemma}

\begin{proof}
For simplicity we assume that the grid $\Grid$ is $\N^2$, i.e., the
period of the grid is 1 and the offset $\offset$ is zero; it is straightforward to generalize to
larger grid periods. With these assumptions, it must be that $A_0$ and $B_0$ are nonnegative integers, which slightly simplifies base cases of our construction.

The CRN input is represented as the initial counts of species $X_1$
and $X_2$, and $\bx = (x_1,x_2)$ denotes the counts of $X_1$ and $X_2$ that have been consumed at any time.
Rather than producing output $f(\x)$ directly upon consumption of
$\x$, our CRN produces $\vp_A(\x)$ copies of a species $Z_1$ and
$\vp_B(\x)$ copies of a species $Z_2$, effectively computing the mapping $M$ described above. Note that $\vp_A(\x)$ and $\vp_B(\x)$ are nonnegative integers by Lemma \ref{lem:well_defined_outside_wedge}.
The CRN works backwards from the quantities $\vp_A(\x)$
and $\vp_B(\x)$ to reconstruct
$f(\x)$.  Roughly this is possible because $f(\x)$ is ``almost'' the
min of $\vp_A(\x)$ and $\vp_B(\x)$, and min is easy to compute.  More
precisely, we can assume that
$f(\x) = \min\{\vp_A(\x),\vp_B(\x)\} - d_i$, where $d_i$ is
determined by the fissure line $L_i$ on which $\n$ resides, and
$d_i=0$ if $i$ is not on a fissure line.  
In addition to the input, a leader $L$ is also present initially. Other
CRN molecules (not initially present) represent a state $[\lx,\lz,d]$ containing three components; we explain the components later.  
Our CRN has three types of reactions: {\em $Z$-producing},
{\em $Z$-consuming}, and {\em $Y$-producing} reactions. The first $Z$-producing reaction handles the base case, producing $(\vp_A(0,0), \vp_B(0,0))= (A_0,B_0)$:
\[
L \rightarrow L' + A_0 Z_1 + B_0 Z_2.
\]
The remaining two $Z$-producing reactions consume $X_1$ and $X_2$ while producing $Z_1$ and $Z_2$. 
If $L_i$ is the line containing $\bx$, the first state component, 
$\lx$, keeps track of $i \mod 2K-1$, where
$K = k + d_{\max}$.
If $i$ is in the range
$[-K+1,K-1]$ then $\lx$ uniquely determines $i$.
For convenience in what follows, we consider $\lx$ to be in the range
$[-K+1,K-1]$ rather than $[0,2K-1]$. The reactions are as follows, where $*$ represents any state component value that is unchanged as a result of the reaction:
\[
\begin{array}{ll}
{[\lx,*,*]} + X_1 & \rightarrow [\lx+(A_1-B_1) \mod 2K-1,*,*] + A_1 Z_1 + B_1 Z_2\\
{[\lx,*,*]} + X_2 & \rightarrow [\lx+ (A_2-B_2) \mod 2K-1,*,*] + A_2 Z_1 + B_2 Z_2
\end{array}
\]

We next describe the $Z$-consuming reactions.  These reactions update
the remaining two components of the state to keep track of which fissure or boundary line contains
$M^{-1}({\bf z})$, where ${\bf z} = (z_1,z_2)$ denotes the counts of
$(Z_1,Z_2)$ that have been consumed at any time.
The reactions also track what is the
{\em deficit}, i.e., the difference between the ``true'' output
$f(M^{-1}({\bf z}))$ and the current output $y$, i.e., number of copies of species $Y$ that has been actually produced so far.
Formally, all reactions maintain the following {\em state invariant}:
if after any reaction the state is $[\lx,\lz,d]$ then
\begin{enumerate}
\item
$\lz$ is the index of the
boundary or fissure line $L_{\lz}$ that contains $M^{-1}({\bf z})$, 
and $\lz$ is in the range $-K+1 \le \lz \le K-1$; and
\item
$d = f(M^{-1}({\bf z})) -y$ is the deficit in the number of $y$'s produced, and is in the  finite range
$-d_{\max} \le d \le   2d_{\max}+1$, where
$d_{\max} = \max \{ d_i \;|\; -k < i < k \}$. \\
\end{enumerate}
$Z$-consuming reactions of the first type handle the base case when $\n = (0,0)$:
\[
L' + A_0Z_1 + B_0Z_2 \rightarrow [A_0-B_0,A_0-B_0, \min\{A_0,B_0\} - d_{A_0-B_0}]. \\
\]
$Z$-consuming reactions of the second type consume a copy of $Z_1$ 
and reactions of the third type consume a copy of $Z_2$.  
Upon consumption, the state
components are updated to ensure that the state invariant holds.
\[
\begin{array}{lll}
{[\lx,\lz,d]} + Z_1 & \rightarrow [\lx,\lz+1,d+d_{\lz}^+], & -K < \lz < k, d \le d_{\max} \mbox{ and}\\
{[\lx,\lz,d]} + Z_2 & \rightarrow [\lx,\lz-1,d+ d_{\lz}^-], & -k < \lz < K, d \le d_{\max}, \\
\end{array}
\]
where
\[
\begin{array}{ll}
d_{\lz}^+ = \left\{
\begin{array}{ll}
d_{\lz} - d_{\lz+1}, &\lz \ge 0  \\
 d_{\lz} - d_{\lz+1} + 1, &\lz < 0.
\end{array} 
\right.
& \mbox{ and }
d_{\lz}^- = \left\{
\begin{array}{ll}
 d_{\lz} - d_{\lz-1} + 1, &\lz \ge 0  \\
d_{\lz} - d_{\lz-1}, &\lz < 0.
\end{array}
\right.
\end{array}
\]
The deficit $d$ can never exceed $2d_{\max}+1$ since the reactions are only applicable when $d\le d_{\max}$ and $d$ can increase by at most $d_{\max}+1$.

The $Y$-producing reactions produce output molecules of species $Y$, while maintaining the state invariant above,
and ensuring that at the
end of the computation the number of $Y$s produced equals $f(\n)$.
The first $Y$-producing reaction produces $d-d_{\max}$ copies of $Y$
when $d$ becomes greater than $d_{\max}$. 
\[
\begin{array}{l}
[*,*,d] \rightarrow [*,*,d_{\max}] + (d - d_{\max})Y, \mbox{ if } d > d_{\max}.
\end{array}
\]                                                
Before describing the remaining $Y$-producing reactions, we
describe some properties of the system of reactions above.  We say that $Z$-consumption {\em stalls} if none of the $Z$-consuming reactions are ever applicable again. Let ${\bf z}_s = (z_{1s},z_{2s})$ be the counts of $(Z_1,Z_2)$ consumed when $Z$-consumption stalls
(${\bf  z}_s$ is independent of the order in which the reactions happen). The $Y$-producing reaction above ensures that
the $Z$-consuming reactions are never stalled
because $d$ becomes too large.  Also, the $Z$-consuming reactions don't stall if $\lz$ is a fissure line and another $Z_1$ is or will eventually be available (and similarly if another $Z_2$ is or will eventually be available), because
$\l_z$ changes by 1 upon consumption of $Z_1$ and so is still less than $K$.

Stalling happens when and only when one of the following (exclusive) cases arise. (i) All copies of both $Z_1$ and $Z_2$ have been consumed and no more will ever be produced, so
${\bf z}_s = (\vp_A(\n),\vp_B(\n))$.
(ii) All copies of $Z_2$ have been consumed and no more will ever be produced, so $z_{2s}=\vp_B(\n)$ but $z_{1s} < \vp_A(\n)$. In this case, $M^{-1}({\bf z}_s)$ is on a lower boundary line.
To see why, note that if $M^{-1}({\bf z}_s)$
were on a fissure or upper boundary line, then the
$Z$-consuming reaction that consumes $Z_1$ would eventually be applicable, because $\lz$ is in the proper range and at least one copy of $Z_1$ has yet to be consumed.
(iii) All copies of $Z_1$ have been consumed and no more will ever be produced, so $z_{1s} = \vp_A(\n)$, but $z_{2s} < \vp_B(\n)$.
In this case, the line $L_{\lz}$ containing $M^{-1}({\bf z}_s)$ must be an upper boundary line.

\begin{claim}
\label{claim:equal}
$f(M^{-1}({\bf z}_s)) = f(\n)$.
\end{claim}

\begin{proof}
This is trivial in case (i) when all $Z$s have been consumed and no more will be produced, since $M^{-1}({\bf z}_s)=\n$.
Consider case (ii) (case (iii) is similar). Then 
$M^{-1}({\bf z}_s) = M^{-1}(z_{1s},\vp_B(n))$, $z_{1s} < \vp_A(\n)$, and the line containing $M^{-1}({\bf z}_s)$ is a lower boundary line.
By Claim \ref{claim:M-increasing},
$\n$ must be in $\dom_B$, because
$\n = M^{-1}(\vp_A(\n),\vp_B(\n))$ and $\vp_A(\n) > z_{1s}$.
Therefore,
$f(M^{-1}({\bf z}_s)) = \vp_B(M^{-1}({\bf z}_s)) = \vp_B(\n) = f(\n)$.\qedhere
\end{proof}

We now return to the last three reactions of the CRN, which are $Y$-producing reactions; we will number them so that we can reference them later and refer to them as {\em deficit-clearing} reactions. The next reaction clears a positive deficit when 
both $\bx$ and $M^{-1}(\z)$ lie on the same fissure line:
\begin{equation}
\label{eqn:y2}
\begin{array}{l}
[\lx,\lz,d] \rightarrow [\lx,\lz,0] + dY 
\mbox{ if } -k < \lx = \lz < k.
\end{array}
\end{equation}
The last two $Y$-producing reactions clear the deficit 
if $M^{-1}({\bf z})$ is on a lower boundary line and for some nonnegative integer $r$, $M^{-1}(\z+(r,0))$ is
on a line $L_l$ with $l = \lx \mod 2K-1$. If such an $r$ exists, let $r_1$ be the smallest such integer
and add the following reaction:
\begin{align}
{[\lx,\lz,d]} + r_1 Z_1 & \rightarrow [\lx,\lz,0]  + r_1 Z_1 + dY, &\lz\ge k, \lz+r_1 = \lx \mod 2K-1, d >0. \label{eqn:y3}
\end{align}
We add a similar reaction when the line $L_{\lz}$ containing
$M^{-1}(\z)$ is an upper boundary line, when an similarly-defined $r_2$
exists:
\begin{align}
{[\lx,\lz,d]} + r_2 Z_2 & \rightarrow [\lx,\lz,0]  + r_2 Z_2 + dY, &\lz\le -k,\lz+r_2 = \lx \mod 2K-1, d >0. \label{eqn:y4}
\end{align}
This completes the description of the CRN. We need one more claim in order to complete the proof of the lemma:

\begin{claim}
\label{claim:nonnegative-deficit}
When $Z$-consumption stalls, the deficit is nonnegative.
\end{claim}

The proof is found in the appendix. To complete the proof, we argue that once $Z$-consumption stalls, some  deficit-clearing reaction will eventually be applicable, ensuring that the output eventually produced is $f(\n)$. If $M^{-1}({\bf z}_s)$ is on a fissure line then $M^{-1}({\bf z}_s)$ must equal $\n$, in which case  $Y$-producing reaction (\ref{eqn:y2}) is applicable.  If $M^{-1}({\bf z}_s)$ is on a boundary line then 
either (\ref{eqn:y3}) or (\ref{eqn:y4}) will be applicable once all inputs are consumed, since for some $r$, either $M^{-1}({\bf z}_s + (r,0)) = \n$ or
$M^{-1}({\bf z}_s + (0,r)) = \n$. Thus in all cases some $Y$-producing reaction eventually clears the deficit, ensuring that the output produced is $f(\n)$.
\end{proof}

\begin{lemma} (Stitching Lemma)
\label{lem:stiching}
Let $f$ be an increasing function. If $f:\N^2 \rightarrow \N$ is a finite combination of output-oblivious functions whose domains are grids, then $f$ is output-oblivious. Also if $f$ is the min of a finite number of output-oblivious functions then $f$ is output-oblivious.
\end{lemma}

\begin{proof}(Sketch)
Let $f$ be a finite combination of output-oblivious functions, say $f_1, f_2, \ldots, f_m$, whose domains are grids. We first describe the construction for the case that the domain $\dom_j$ of $f_j$ is a two-way-infinite grid for all $j, 1\le j \le m$. Let the
offset of the $j$th grid be $\offset_j = (o_{j,1}, o_{j,2})$.  On input $\n$, our CRN $\C$ first produces $m$ distinct ``inputs'' $\nj{j} \in \N^2$ such that $\n \le \nj{j}$ and $\n = \nj{j}$ if $\n \in \dom_j$.  From these, $\C$ produces $m$ ``outputs'' $y_j = f_j(\nj{j})$, using CRNs $\C_j$ for each $f_j$. Finally, $\C$ produces $y = \min\{y_1, \ldots, y_m\}$.

To see that such a $\C$ is correct, i.e., that $y = f(\n)$, note that if $\n \in \dom_j$ then $y_j = f_j(\nj{j}) = f_j(\n) = f(\n)$, since $\n = \nj{j}$, and if $\n \not \in \dom_j$ then $y_j = f_j(\nj{j}) \ge f(\n)$, since $f$ is increasing and $\nj{j} \ge \n$.
Thus $f(\n) = \min \{y_1, \ldots, y_m \}= y$.  The details of producing the $\n^{(\vb{j})}$s and the output are in Appendix \ref{app:proof_stiching}. 

When $f$ is the min of a finite number of output-oblivious functions, say $f_1, f_2, \ldots, f_m$, we can similarly stably compute each $f_i$ using an output-oblivious CRN $\C_i$ such that the species for each $\C_i$ are distinct, and then take the min of the outputs as the result.
\end{proof}

We complete this section by proving the sufficiency (if) direction of Theorem \ref{thm:output_oblivous_iff}. \\

\noindent
{\bf Theorem 
\ref{thm:output_oblivous_iff} (if direction).}
\emph{A semilinear function $f:\N^2\to\N$ is output-oblivious if $f$ is increasing and is either grid-affine or the minimum of finitely many fissure functions.}

\begin{proof}
First suppose that $f$ is grid-affine. Then by Definition \ref{def:grid-affine}, $f$ is a combination of affine functions $f_1,\ldots,f_m$ whose domains
$\Grid_1, \ldots, \Grid_m$ respectively are grids.  By Lemma \ref{lem:partialgridoo}, $f_i: \Grid_i \to \N$ is output-oblivious, $1\le i \le m$.  By Lemma \ref{lem:stiching}, $f$
is output-oblivious.

Otherwise $f$ is the min of finitely many complete fissure functions
$f_1,\ldots,f_m$.  By Definition \ref{def:fissure_function}, each $f_i$ is a combination of partial fissure functions on grids of period $p$, for some $p \in \N$. By Lemma \ref{lem:fisureoo}, each of these partial fissure functions is
output-oblivious. By Lemma \ref{lem:stiching}, each $f_i$ is output-oblivious
and also $f$ is output-oblivious.
\end{proof}

\section{Proof of Necessity in Theorem \ref{thm:output_oblivous_iff}}
\label{sec:necessity_proof}

In this section we prove that if a function is output-monotonic, then it is either grid-affine or the minimum of finitely many fissure functions. 
In Section \ref{sec:impossibility_lemmas} we describe two conditions on a function which ensure that it is not output-oblivious. In Section \ref{sec:properties_semi_affine} we show two technical results which are needed in Section \ref{sec:necessity-proof-subsec}, which contains the proof of necessity. 
The arguments made in this section pertain to output-monotonic functions, allowing us to both characterize this set of functions and output-oblivious functions since any output-oblivious function is clearly output-monotonic. 

\subsection{Impossibility Lemmas}
\label{sec:impossibility_lemmas}

The results of this section use Dickson's Lemma:

\begin{lemma}(Dickson's Lemma \cite{dickson1913})
Any infinite sequence in $\N^k$ has an infinite, non-decreasing subsequence.
\end{lemma}

\begin{lemma}
\label{lem:impossibility1}
Let $f:\N^2\to\N$ be a semiaffine function. Suppose that $f = \vp_i$ on $\dom_i$ and $f = \vp_j$ on $\dom_j$, where $\dom_i$ and $\dom_j$ lie on the same grid,
$\dom_i$ is a wedge domain and for some two-way-infinite line $L$ in $\dom_j$,
$\vp_j > \vp_i$ on $L$. Then $f$ cannot be stably computed by an output-monotonic CRN. 
\end{lemma}

\begin{proof}
Suppose to the contrary that CRN $\C$ stably computes $f$. Either $\dom_i$ is counter-clockwise to $\dom_j$ or vice versa. Assume it is the latter; the proof is similar if the orientation
of the domains is reversed. 

Let $\{\vb{p}_k\}_{k\in\N}$, be an infinite sequence of points 
in $\dom_i$ that are strictly increasing in the first dimension, and form a line which is not parallel to $L$ (this is possible since $\dom_i$ is a wedge domain). 
Let $\vb{c}_k$ be a stable configuration reached on a computation of $\C$ on input $\vb{p}_k$.
Applying Dickson's Lemma, choose an infinite subsequence of
$\{\vb{c}_k\}$ and renumber so that 
$\vb{c}_{k} \le \vb{c}_{k+1}$ for all $k$. 
Let $\n_{1}, \n_{2}, \ldots$  be another strictly increasing sequence in $\N^2$
such that
$\vb{p}_k + \n_{k} \in L$ and
$\vb{p}_{k'} + \n_{k} \in \dom_i$ for $k' > k$. Such a sequence exists because $\{\vb{p}_k\}$ is not parallel to $L$, and $L$ is two-way-infinite. 

If $\C$ is correct, then on input $\vb{p}_k + \n_{k}$ some execution sequence of $\C$ first  reaches configuration $\vb{c}_k$, which has $\vp_i(\vb{p}_k)$ copies of $Y$, and then outputs $\vp_j(\vb{p}_k + \n_{k}) - \vp_i(\vb{p}_k)$ additional $Y$s (since $\vb{p}_k+\n_k\in\dom_j$). Let $k' > k$. On input
$\vb{p}_{k'} + \n_{k}$ (which is in $\dom_i$), $\C$ can output a number of $Y$s equal to:
\[
\vp_i(\vb{p}_{k'}) + \vp_j(\vb{p}_k + \n_{k}) - \vp_i(\vb{p}_k).
\]
This occurs when $\C$ first produces stable output 
$\vp_i(\vb{p}_{k'})$ while consuming input $\vb{p}_{k'}$ and reaching
configuration $\vb{c}_{k'}$. Since $\vb{c}_{k'}\geq \vb{c}_k$, it can then follow the same execution that it would follow from $\vb{c}_k$
to produce $\vp_j(\vb{p}_k + \n_{k}) - \vp_i(\vb{p}_k)$ copies of $Y$ (since all the necessary species are available). 
However, since $\vp_j > \vp_i$ on $L$, the number of $Y$s produced is greater than
\[
\vp_i(\vb{p}_{k'}) +  \vp_i(\vb{p}_k + \n_{k}) - \vp_i(\vb{p}_k) = \vp_i(\vb{p}_{k'} + \n_{k})
\]
(the equality here follows since $\vp_i$ is affine).  Thus too many $Y$'s can be output by $C$ on input $\vb{p}_{k'}+\n_{k}$, and so $C$
cannot be output-monotonic.
\end{proof}

\begin{lemma}
\label{lem:impossibility2}
Let $f:\N^2\to\N$ be a semiaffine function. Suppose that $f = \vp_i$ on $\dom_i$ and $f = \vp_j$ on $\dom_j$, where $\dom_i$ and $\dom_j$ are wedge domains on the same grid
$\G$ such
that (i) $\vp_i = \vp_j$ on $\N^2$ and (ii) there is a two-way-infinite line $L\subset\G$ separating $\dom_i$ and $\dom_j$, with $\vp_L < \vp_i
\;(=\vp_j)$ on $L$.  Then $f$ cannot be stably computed by an
output-monotonic CRN.
\end{lemma}

The proof is similar to the previous lemma and is in Appendix~\ref{app:impossibility2}. 

\subsection{Properties of Increasing Semiaffine Functions}
\label{sec:properties_semi_affine}
Here we show several useful properties of increasing semiaffine functions $f:\N^2\to\N$. 
The proofs are in the Appendices \ref{app:technical1proof} and \ref{app:technical2proof}. 
 For a partial affine function $\vp_i$ with domain $\dom_i$, write $\vp_i(\n)=\la \a_i,\n\ra+a_{0,i}$ (where $\la\cdot,\cdot\ra$ denotes the standard inner product).  We say that a line $L$ is a \emph{constant distance} from $\dom_i$ if there exists some constant $K$ such that for all $\n\in L$ there is some $\vb{c}(\n)$ with $-(K,K)\leq \vb{c}(\n)\leq (K,K)$ and $\vb{n}+\vb{c}(\n)\in\dom_i$.

\begin{lemma}
\label{lem:line_offset}
Let $f:\N^2\to\N$ be an increasing semiaffine function. Suppose that $f = \vp_i$ on $\dom_i$ and $f = \vp_j$ on $\dom_j$, where $\vp_i$ and $\vp_j$ are partial affine functions, $\dom_i$ and $\dom_j$ are two-way infinite domains,
and some line $L$ in $\dom_j$ is a constant distance from $\dom_i$.
 Then there exists some $\kappa\in\Q$ such that $\la \a_i,\n\ra = \la \a_j,\n\ra + \kappa$ for all $\n\in L$. 
\end{lemma}

\begin{lemma}
\label{lem:equal_coefficients}
Let $f:\N^2\to\N$ be an increasing semiaffine function. Suppose that $f = \vp_i$ on $\dom_i$ and $f = \vp_j$ on $\dom_j$, where $\vp_i$ and $\vp_j$ are partial affine functions.
If there exist two non-parallel lines $I,L\subset \dom_i$ which are both a constant distance from $\dom_j$, then $\a_i=\a_j$. 
\end{lemma}

What follows is an easy consequence of the previous lemma. Its proof is also found in Appendix \ref{app:technical2proof}. 

\begin{lemma}
\label{lem:equal_functions}
Let $f:\N^2\to\N$ be an increasing semiaffine function.
If there exists a two-way-infinite domain $D$ on which $\vp_i$ and $\vp_j$ are well-defined such that $f(\n)=\vp_i(\n)=\vp_j(\n)$ on all $\n\in D$, then $\vp_i=\vp_j$. 
\end{lemma}

\subsection{Proof of Necessity}
\label{sec:necessity-proof-subsec}
We first provide the main argument in the proof of necessity---that the desired result holds on individual grids. 

\begin{lemma}
\label{lem:only_if_partial_fissures}
Let $f:\N^2\to\N$ be an increasing semiaffine function with period $p$. Then $f$ is output-monotonic only if for any large enough offset $\offset$, $f\on_{\G_\p+\offset}$ is either affine or the minimum of finitely many partial fissure functions.
\end{lemma}

\begin{proof}
Fix a representation of $f$ as a semiaffine function with period $p$. Choose $\offset$ large enough so that, if $\Grid = \Grid_p + \offset$, no domains of $f$ that are points or one-way-infinite domains overlap	$\Grid$, and also no two-way-infinite domains of $f$ cross. Assume that $f\on_\G$ is not affine; we need to show that $f\on_\Grid$ is the minimum of partial fissure functions. Assume that the representation of $f$ on $\Grid$ minimizes the number of wedge domains.

Consider all line domains in the representation of $f\on_\Grid$, plus all two-way-infinite lines that define the top or bottom boundaries of wedge domains. Partition these lines into maximal sets 
of parallel lines. For each such set $s$, we define a function $f_s:\Grid \to \N^2$, and show in Claim \ref{claim:partial_fissure} that each $f_s$ is a partial fissure function. Let ${\cal S}$ be the set of all of the sets $s$ of parallel lines. We show in Claim \ref{claim:min_fissures} that $f\on_\Grid = \min_{s \in {\cal S}} f_s$, completing the proof of Lemma \ref{lem:only_if_partial_fissures}. \\

\noindent
{\bf Definition of $f_s$.} Fix any maximal set $s$ of parallel lines. Without loss of generality we assume that $s$ has at least three lines (we can add additional domains to $f$'s representation if needed to ensure this). Some line of $s$, say $L_{A,s}$, defines the lower boundary of a wedge domain; we assume that this is top line of $s$ (we can remove line domains of $s$ above $L_{A,s}$ from $f$'s representation if needed to ensure this). Let $f = \vp_{A,s}$ on this wedge domain, where $\vp_{A,s}$ is a partial affine function. In what follows, we drop the subscript $s$ when referring to these and other domains and functions. Let $\dom_A$ be the wedge of points of $\Grid$ that lie on or above $L_A$. (Note that $\dom_A$ may not be a domain of $f\on_\Grid$ since lines from some other set $s_j$ may lie above the lines of $s$.) Similarly, we can assume that the bottom line, say $L_B$, of $s$ defines the upper boundary of a wedge domain. Let $f = \vp_B (=\vp_{B,s})$ on this wedge domain and let $\dom_B$ be the wedge of points of $\Grid$ that lie on or below $L_B$. 
We can assume without loss of generality (by adding more lines if necessary and further adjusting which lines are $L_A$ and $L_B$) that any line $L$ that lies between, and is parallel to, $L_A$ and $L_B$ is a (possibly empty) domain of $f\on_G$.
Number these lines $L_1, L_2, \ldots$, say from top to bottom, and let $\vp_i$ be the partial affine function that agrees with $f$ on line $L_i$.

Let $f_s:\G\to\N$ be the following function associated with set $s$:
\begin{equation}
\label{eq:partial_fissure_construction}
    f_s(\n)=\begin{cases}
\vp_{A}(\n),&\text{if }\n\in \dom_A, \\
\vp_{i}(\n),&\text{if }\n\in L_i, 2 \le i \le |s|-1, \\
\vp_{B}(\n),&\text{if }\n\in \dom_B.
\end{cases}
\end{equation}

See Figure \ref{fig:min_fissures_construction} for an illustration of the construction. 

\begin{figure}
    \centering
    \includegraphics[scale=0.3]{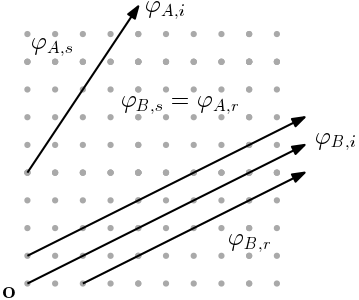}
    \caption{An example of the construction in Lemma \ref{lem:only_if_partial_fissures}. Here $\S=\{s,r\}$, where $s$ is the topmost line and $r$ is the bottommost three lines. In the construction we add two more lines to $s$ since $|s|=1$.  }
    \label{fig:min_fissures_construction}
\end{figure}

\begin{claim}
\label{claim:partial_fissure}
$f_s$ is a partial fissure function.
\end{claim}

\begin{proof}
Note that the lines $L_i$ are parallel to, and have constant distance from, both $L_A$ and $L_B$.  By Lemma \ref{lem:line_offset}, for some $d_{i,A} \in \Q$ we have that $\vp_{i}(\n)=\vp_{A}(\n)-d_{i,A}$ for all $\n$ on line $L_i$. That is, $L_i$ is the set of all $\n$ such that $\vp_A(\n)-\vp_i(\n)=d_{i,A}$. Since $\vp_i(\n)$ and $\vp_A(\n)\in \N$ for all $\n$ in $\Grid \cap L_i$, $d_{i,A}$ must be in $\Z$. Also, since $f$ is output-monotonic, Lemma \ref{lem:impossibility1} shows that $\vp_{i} \le \vp_{A}$ on $L_i$ and so $d_{i,A}$ must be in $\N$. By reasoning similar to that in the last paragraph, $\vp_{i}(\n)=\vp_{B}(\n)-d_{i,B}$ for all $\n$ on line $L_i$, for some $d_{i,B} \in \N$. It follows that $L_i$ is the set of points $\n \in \Grid$ for which $\vp_A(\n) - \vp_B(\n) = k$, for some $k \in \Z$.
It follows 
(by potentially adding yet more lines if necessary so that the number of lines $L_i$ that lie above the line $\vp_A(\n) - \vp_B(\n)=0$ is equal to the number of lines of $s$ that lie below the line $\vp_A(\n) - \vp_B(\n)=0$) that we can represent $f_s(\n)$ as in Definition \ref{def:fissure_function}. 

It remains to show that
$\vp_A<\vp_B \mbox{ on} \dom_A \mbox{ and } \vp_B<\vp_A \mbox{ on} \dom_B$.
Suppose that there exists a point $\n\in \dom_A$ such that $\vp_{A}(\n)\ge\vp_{B}(\n)$. If there are only finitely many such points in any dimension, then we may disregard them by taking $\offset$ sufficiently large. Hence, we may assume that if there is one such point then there are infinitely many and they form a two-way-infinite domain, $D$. 
Let $E=\{\n\in D:\vp_A(\n)>\vp_B(\n)\}$ and $F=\{\n\in D:\vp_A(\n)=\vp_B(\n)\}$. Note that $E\cup F=D$. We consider three cases. 
If $E$ and $F$ are both one-way-infinite, then one is a horizontal line and the other a vertical line. As above, we may disregard these points by taking $\offset$ large enough. 
If $E$ is two-way-infinite we can find a two-way-infinite line $L\subset \dom_A$ such that $\vp_{A}>\vp_{B}$ on $L$. By Lemma \ref{lem:impossibility1}, this contradicts the fact that $f$ is output-monotonic. 
Otherwise, $F$ is two-way-infinite.
By Lemma \ref{lem:equal_coefficients} we see that $\vp_A=\vp_B$. Furthermore, for some $i$, we have $\vp_i<\vp_A$. Otherwise, $d_{i,A}=0$ for all $i$ and consequently the number of wedge domains of $f\on_\G$ can be reduced by merging the domains $\dom_A$, $\dom_B$ and the $L_i$ contradicting our assumption that the representation of $f$ on $\Grid$ minimizes the number of wedge domains. However, $\vp_A$, $\vp_B$ and $\vp_i$ then meet the conditions of Lemma \ref{lem:impossibility2}, a contradiction. The proof that $\vp_B<\vp_A$ on $\dom_B$ is similar.
\end{proof}

\begin{claim}
\label{claim:min_fissures}
$f\on_\Grid = \min_{s \in {\cal S}} f_s$.
\end{claim}
\begin{proof}
 Let $\n\in\dom(f_s)$, so that $f\on_\G(\n)=f_s(\n)$. Suppose to the contrary that $f_{s'}(\n)<f_s(\n)$ for some $s' \in {\cal S}$. We consider the case where the lines of $s'$ lie above those of $s$; the case where the lines of $s'$ lie below those of $s$ is similar. Hence $f_{s'}(\n)=\vp_{A,s'}(\n)$. We consider three distinct cases based on the partitioning of $\dom(f_s)$. 

\begin{enumerate}
    \item $\n\in\dom_{A,s}$ so $f_s(\n)=\vp_{A,s}(\n)$. Applying the same argument as in the proof of Claim \ref{claim:partial_fissure}, since $\vp_{A,s'}(\n)<\vp_{A,s}(\n)$, there is a line $L\subset\dom_{A,s}$ such that $\vp_{A,s'}<\vp_{A,s}$ on $L$. Thus, $\vp_{A,s'}$ and $\vp_{A,s}$ and their respective domains meet the condition of Lemma \ref{lem:impossibility1}, contradicting the fact that $f$ can be stably computed by an output-monotonic CRN. 
    \item $\n\in \dom_{i,s}$ for some $i$ with $2\leq i\leq |s|-1$, so $f_s(\n)=\vp_{i,s}(\n)$. Since $f_{s}$ is a fissure function, we have $\vp_{i,s}(\n)=\vp_{A,s}(\n)-d_i$, for some $d_i \in \N$, hence $f_s(\n)=\vp_{i,s}(\n)\leq \vp_{A,s}(\n)\leq \vp_{A,s'}(\n)=f_{s'}(\n)$.
    Clearly then, it cannot be the case that $f_{s'}(\n)<f_s(\n)$.  
    \item $\n\in\dom_{B,s}$. This is similar to Case 1. 
\end{enumerate}
Since we get a contradiction in all three cases, we conclude that $f_s(\n)$ must be less than or equal to $f_{s'}(\n)$ for all $s' \in {\cal S}$, and the claim is proved.
\end{proof}

This completes the proof of Lemma \ref{lem:only_if_partial_fissures}.
\end{proof}

Lemma \ref{lem:only_if_partial_fissures} describes properties of
increasing semiaffine functions on one grid. We now turn to properties of such functions across grids. The proof of the next lemma builds on that of Lemma \ref{lem:only_if_partial_fissures} and uses Lemma \ref{lem:equal_coefficients} to relate $f$ across grids.

\begin{lemma}
\label{lem:cross_grids}
Let $f:\N^2\to\N$ be an output-monotonic semiaffine function with
period $p$. Let $\offset$ be large enough that
for all $\offset' \ge \offset$,  $f\on_{\G_\p+\offset'}$
is either affine or the minimum of finitely many partial fissure functions. If $f$ is affine on
$\Grid=\Grid_p+\offset$, then $f\on_{\Grid_p + \offset'}$ must also be affine for all $\offset' \ge \offset$, and thus $f$ is grid-affine. Otherwise $f$ is the
minimum of finitely many fissure functions.
\end{lemma}

\begin{proof}
Let $\Grid'=\Grid_p+\offset'$.  Suppose $f\on_\Grid$ is affine but that $f\on_{\Grid'}$ is not. Write $f\on_\Grid(\n)=\la\a,\n\ra+a_0$.  By Lemma \ref{lem:only_if_partial_fissures}, $f\on_{\Grid'}$ is the minimum of partial fissures functions of the same form as in \eqref{eq:partial_fissure_construction}. Fix one of these functions $f_s$, and let $\vp_A$ and $\vp_B$ be as in \eqref{eq:partial_fissure_construction}. As in the construction, $f\on_{\Grid'}=\vp_A$ on some wedge domain $D_A\subset\dom_A$ ($D_A$ might not be $\dom_A$ for technical reasons discussed in the previous proof)  and similarly, $f\on_{\Grid'}=\vp_B$ on $D_B\subset\dom_B$. Now, $D_A$ and $\Grid$ meet the conditions of Lemma \ref{lem:equal_functions}, so $A_1=a_1$ and $A_2=a_2$ where $\vp_A(\n)=A_1n_1+A_2n_2+A_0$. Likewise, however, we see that $a_1=B_1$ and $a_2=B_2$, where $\vp_B(\n)=B_1n_1+B_2n_2+B_0$. But this contradicts the definition of a fissure function, since we should have $A_1>B_1$ and $A_2<B_2$. Thus $f\on_{\Grid'}$ must also be affine.  

Conversely, suppose $f\on_\Grid$ is the minimum of the partial fissure functions $\{f_s\}_{s\in\S}$. Our goal is to show that $f\on_{\Grid'}$ can be written as the minimum of partial fissures $\{f_s'\}_{s\in \S}$, where $f_s'$ corresponds to $f_s$ but on the grid $\Grid'$. This will demonstrate that each $f_s$ is part of a complete fissure function, and that $f_{\{\n\geq\offset\}}$ is the minimum of these. 

Similarly to above fix one of the functions $f_s$ in the representation of $f\on_{\Grid}$ \eqref{eq:partial_fissure_construction} and let $f\on_\Grid=\vp_A$ on $D_A$ etc.
The following general construction will be useful. Let $E\subset \Grid$ be a line or wedge domain on the grid $\Grid$. $E$ can be described by the intersection of threshold sets with $\G$. We let $E'$ refer to the set defined by the same threshold sets but intersected with $\Grid'$. Intuitively, $E$ and $E'$ are the same set defined on different grids. 
Let $D_s\subset\Grid$ be the wedge domain on which $f\on_\Grid=f_s$ (recall that $f_s$ is defined on all of $\N^2$, but it does not agree with $f\on_\Grid$ everywhere.
As above, using Lemma \ref{lem:equal_functions}, we see that $f\on_{\Grid'}=\vp_A+\kappa_A$ on $D_A'$ for some constants $\kappa_A$. Similarly, $f\on_{\Grid'}=\vp_B+\kappa_B$ on $D_B'$ and, using Lemma \ref{lem:equal_coefficients}, $f\on_{\Grid'}=\vp_i+\kappa_i$ on $L_i'$ for $2\leq i\leq |s|-1$. Thus we can write 
\[f\on_{\Grid'\cap D_s'}(\n)=
\begin{cases}
\vp_{A}(\n)+\kappa_A,&\text{if }\n\in D_A', \\
\vp_{i}(\n)+\kappa_i,&\text{if }\n\in L_i', 2 \le i \le |s|-1, \\
\vp_{B}(\n)+\kappa_B,&\text{if }\n\in D_B'.
\end{cases}
\]
Moreover, this function must obey the properties of a fissure function, otherwise it would not be output-monotonic by Lemma \ref{lem:only_if_partial_fissures}. \end{proof}

We complete this section by proving the necessity (only if) direction of Theorem \ref{thm:output_oblivous_iff}, strengthening it slightly so that it applies also to output-monotonic functions. \\

\noindent
{\bf Theorem 
\ref{thm:output_oblivous_iff} (only if direction).}
\emph{A semilinear function $f:\N^2\to\N$ is output-monotonic only if $f$ is increasing and is either grid-affine or the minimum of finitely many fissure functions.}

\begin{proof}
By Lemma \ref{lem:semi_affine_function}, $f$ is a semiaffine function
with period $p$. By Lemmas \ref{lem:only_if_partial_fissures} and \ref{lem:cross_grids}, $f$ is either grid-affine or the minimum of
fissure functions.
\end{proof}

\section{Conclusions and Future Work}

Here we have characterized the class of functions $f:\N^2 \to \N$ that can be stably computed by output-oblivious and output-monotonic
stochastic chemical reaction networks (CRNs) with a leader.  A natural next step for future work is to generalize the result to functions $f:\N^k \to \N$ for $k >2$; we are optimistic that many of the building blocks that we introduce here for the two-dimensional case will generalize to the multi-dimensional case.

Another natural question is to determine what can be computed when
there is no leader.  By similar reasoning to that of Chalk et
al.~\cite{Chalk-etal-2018} for mass-action CRNs, such functions must
be super-additive, but whether only the super-additive semilinear
functions have output-oblivious stochastic CRNs remains to be
determined.

Yet another direction for future work is to determine whether, for some functions, there is a provable gap between the time needed to stably compute the functions with an output-oblivious CRN and the time needed by a CRN that is not restricted to be output-oblivious. Further directions are to better understand output-oblivious CRN computation when errors are allowed, and whether it is possible to "repair" a CRN that is not output-oblivious so that composition is possible.

\bibliography{ref.bib}

\newpage
\appendix

\noindent{\Large \bf Appendix}

\section{Proof of Lemma \ref{lem:semi_affine_characterization}}
\label{app:set_characterization}

Here we prove Lemma \ref{lem:semi_affine_characterization}. To do this, we first prove, via a sequence of lemmas, that every linear set is in fact a finite union of sets of the form $\G\cap \T$, where $\G$ is a grid and $\T$ is a finite intersection of threshold sets, where moreover the grids have the same period. We use several useful properties of semilinear sets due to Angluin et al.~\cite{angluin2007computational}. 

A \emph{horizontal} threshold set has the form $\{\n: n_2\geq r\}$ for some $r\in\Z$, i.e., the set includes all the points above the horizontal line $n_2=r$. A \emph{vertical} threshold set is defined similarly, including all the points to the right of some vertical line. The \emph{slope} of a threshold set $\{\n:\n\cdot \v\geq r\}$ is the slope of its bounding line $\{\n:\n\cdot\v=r\}$, which is $-v_1/v_2$.    

\begin{definition}
\label{def:wedge-domain}
Let $H$ and $V$ be horizontal and vertical threshold sets, respectively. Let $T_U$ and $T_L$ be threshold sets such that the slope of $T_U$ is strictly greater than that of $T_L$. A set $D\subset\N^2$ is a \emph{wedge domain} if it can be written as $\Grid\cap T_U\cap T_L\cap H\cap V$ for some grid $\Grid$. See Figure \ref{fig:wedge_domain}. 
\end{definition}

\begin{figure}
    \centering
    \includegraphics[scale=0.3]{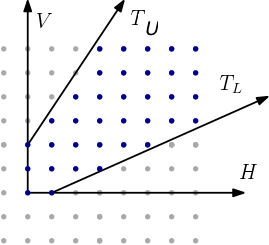}
    \caption{An example of a wedge domain.}
    \label{fig:wedge_domain}
\end{figure}

A \emph{modulo set} is a set of the form $\{\n:\n\cdot\v \equiv r\mod p\}$ for some $\v\in\N^2$, and $r,p\in\N$. 
We call $p$ the \emph{period} of the set. Recall that a \emph{threshold set} has the form $\{\n:\n\cdot\v\geq r\}$ for some $\v\in\Z^2$, $r\in\Z$. 
A \emph{boolean combination} of sets refers to a combination of sets by union and intersection. The following lemma relates semilinear sets to modulo and threshold sets. 

\begin{lemma}[\cite{angluin2007computational}]
\label{lem:angluin_semilinear}
Every semilinear set can be represented as a finite boolean combination of modulo sets and threshold sets. 
\end{lemma}

Next, we make a simple observation concerning the period of grids. 
\begin{lemma}
\label{lem:modulo_period}
Any grid $\Grid=\{(p\alpha_1,q\alpha_2)+(o_1,o_2):\alpha_i\in\N\}$ can be written as a finite union of grids with the same period. 
\end{lemma}
\begin{proof}
Let $k$ be a multiple of both $p$ and $q$. For all $i,j\in\N$ such that $pi<k$ and $qi<k$, let $\offset_{i,j}=(o_1+pi,o_2+qi)$. It is then easily verified that 
\[\Grid=\bigcup_{i,j}\Grid_k+\offset_{i,j}.\qedhere\]
\end{proof}



The next lemma demonstrates that modulo sets are effectively unions of grids in hiding. For example, the modulo set $M=\{\n:2n_1+x_2\equiv 0\mod 3\}$ can be written as $\G_3\cup \G_3+(1,1)\cup\G_3+(2,2)$.

\begin{lemma}
\label{lem:modulo_set_expansion}
Any modulo set can be expressed as the finite union of grids. 
\end{lemma}
\begin{proof}
Let $M=\{\n:\n\cdot\v\equiv a\mod p\}$ be a modulo set. Write $M=\bigcup_{i=0}^\infty M_i$ where $M_i=\{\n:n_1=i\text{ and } v_2n_2\equiv a-v_1i\mod p\}$. Appealing to properties of modular arithmetic, we can write $M_i=\{(i,q\alpha)+(0,v_i):\alpha\in\N\}$ for some $q\in\N$. 

Now, we observe that $\Pi_2(M_i)=\Pi_2(M_j)$ whenever $i\equiv j\mod p$. Indeed, if $x\in\Pi_2(M_i)$, then $v_2x\equiv a-v_1i\mod p\equiv a-v_1j\mod p$ if $i\equiv j\mod p$ since $v_1\in\N$. Hence, $x\in\Pi_2(M_j)$. The other inclusion is identical. This implies that for such $i$ and $j$, $v_i=v_j$ for $i\equiv j\mod p$. For every distinct $r<p$, define 
\begin{equation*}
    F_r = \bigcup_{i\equiv r\mod p}\{(i,q\alpha)+(0,v_i):\alpha_1,\alpha_2\in\N\}=\{(p\alpha_1+r,q\alpha_2+v):\alpha_1,\alpha_2\in\N\},
\end{equation*}
for some $v$, which is a grid. Noting that $\bigcup_{r<p}F_r=\bigcup_{i=0}^\infty M_i=M$ completes the proof. 
\end{proof}

Since the representation of a semilinear set $E$ may also have intersections of modulo sets (according to Lemma \ref{lem:angluin_semilinear}) we must be able to write the intersection of modulo sets as grids as well. We do this by first converting the modulo sets into grids as per Lemma \ref{lem:modulo_set_expansion} and then reasoning about the intersection of grids. The next lemma allows us to write the intersection of these grids as the union of other grids. As an example, consider $A=\{(\alpha_1,3\alpha_2):\alpha_i\in\N\}$ and $B=\Grid_2 + (1,2)$. Here, $A\cap B$ is the grid $\{(2\alpha_1,6\alpha_2)+(1,0):\alpha_i\in\N\}$. 


\begin{lemma}
\label{lem:modular_intersection}
The intersection of two grids can be expressed as the finite union of grids.  
\end{lemma}
\begin{proof}
Let $M=\{(x_1,0)\alpha_1+(0,x_2)\alpha_2+\u:\alpha_i\in\N\}$ and $N=\{(y_1,0)\beta_1+(0,y_2)\beta_2+\v:\beta_i\in\N\}$ be grids. As in the example, set $M_i=\{(x_1,0)i+(0,x_2)\alpha+\u:\alpha\in\N\}$ and do similarly for $N_i$. Then $\bigcup_{i=0}^\infty M_i=M$ and $\bigcup_{i=0}^\infty N_i=N$. Further, $M_i\cap N_j\neq\emptyset$ only if $i=(x_2j+v_1-u_1)/x_1$. 
Fix such $i$ and $j$. We can write $M_i=\{(y_1,0)j+(0,x_2)\alpha +(v_1,u_2) :\alpha\in\N\}$. Now, the $x_2$-projections of $M_i$ and $N_j$ are non-empty whenever $x_2\alpha + u_2=y_2\beta + v_2$. This is a linear diophantine equation that has a solution iff $v_2-u_2|\gcd(x_2,y_2)$. If there is no solution, then $M_i\cap N_j=\emptyset$, implying that $M\cap N=\emptyset$. In this case we're done. Assuming there are solutions, there are infinitely many and they are described by the family $\{(s+ku,t-kv):k\in\Z\}$ where $(s,t)$ is a solution, and $u=y_2/\gcd(x_2,y_2)$, $v=x_2/\gcd(x_2,y_2)$. Let $k$ be minimal such that $s+k u\in\N$. This is therefore the minimum non-negative solution for $\alpha$. It follows that $M_i\cap N_j=C_{ij}=\{(y_1,0)j+(0,x_2)(s+ku)+(v_1,u_2):k\in\N\}=\{(y_1,0)j+(0,x_2u)\alpha + (v_1,u_2+s):\alpha\in\N\}$, and so
\[M\cap N=\bigcup_{i=1}^\infty C_{ij}=\{(y_1,0)\alpha_1+(0,x_2u)\alpha_2+(v_1,u_2+\ell):\alpha_i\in\N\}.\qedhere\]
\end{proof}

Finally, using distributivity of set operations we can prove the main result. 

\begin{lemma}
Every semilinear set can be written as a finite union of sets of the form $\G\cap \T$, where $\G$ is a grid and $\T$ is a finite intersection of threshold sets. Moreover, we may assume that each grid has the same period. \end{lemma}
\begin{proof}
Let $L$ be semilinear. Using Lemma \ref{lem:angluin_semilinear} express $L$ as a finite boolean combination of threshold and modulo sets. Expressing each modulo set as a finite union of grids according to Lemma, \ref{lem:modulo_set_expansion} and then using distributivity of set operations and Lemma \ref{lem:modular_intersection} we can write $L=\bigcup_{j=1}^N \Grid^j\cap \T^j$ for some $N\in\N$ where each $\T^j$ is a finite \emph{intersection} of threshold sets and $\Grid^j$ is a grid.  Finally, applying Lemma \ref{lem:modulo_period}, we may assume that each $\Grid^j$ has the same period. 
\end{proof}

It now remains only to remark that we have in fact proven Lemma \ref{lem:semi_affine_characterization}.

\begin{proof}[Proof of Lemma \ref{lem:semi_affine_characterization}]
Let $L$ be semilinear and using the previous lemma write $L=\bigcup_{j=1}^N \Grid^j\cap \T^j$. Fix $M=\Grid^j\cap \T^j$. If $M$ is finite, then it can be written as the finite union of points. If $\T^j$ includes only the boundary of some threshold set, then $M$ is a line on the grid $\Grid^j$. Otherwise, $M$ is infinite and is not a line. Thus it is two-way infinite. As $n_1$ and $n_2$ grow, there exist unique threshold sets $T_U$ and $T_L$ in $\T^j$ which bound $M$ from above and below, respectively. The effect of the other (finite number of) threshold sets disappear in the limit because they have different slopes. More specifically, threshold sets with a larger slope than $T_U$, and smaller than that of $T_L$ cease to constrain $M$ as the coordinates $n_1$ and $n_2$ grow. Therefore, for some $n_1^*$ and $n_2^*$ large enough, we can write $M=(G^j\cap T_U\cap T_L\cap H\cap V) \cup F$, where $H$ is the horizontal line defined by $n_1^*$, $V$ the vertical line defined by $n_2^*$, and $F=\{(n_1,n_2)\in M: n_1\leq n_1^*, \; n_2\leq n_2^*\}$. Moreover, $F$ is finite since it is constrained above and below by the threshold sets in $T^j$ ($T_U$ and $T_L$ in particular.) 
That is, $M$ can be written as the union of a wedge domain and finitely many points, which completes the proof. 
\end{proof}

\section{Proof of Lemma \ref{lem:semi_affine_function}}
\label{app:semi_affine_functions}

\begin{proof}
From Definition \ref{def:semilinear-classification} (Chen et al.), $f$ can be represented by the partial functions $\{\vp_1,\dots,\vp_m\}$ where $\dom_i$ is linear for each $i$.  
Applying Lemma \ref{lem:angluin_semilinear} and Lemma \ref{lem:modulo_period}, we can write $\dom_i$ as the union of semiaffine sets with the same period, say $p$. Therefore, each grid of period $p$ is covered by a distinct function $f_i$ which is defined by the union of all affine functions whose domain covers this grid. Thus, it follows that $f$ is a periodic combination of partial semiaffine functions on grids. 
\end{proof}

\section{Proof of Lemma \ref{lem:well_defined_outside_wedge}}
\begin{proof}
Let $\vp(\n)=\la \a,n\ra +a_0$ and $\G=\{(x_1\alpha_1,x_2\alpha_2)+\offset:\alpha_i\in\N\}$. Let $\n'\in \G\setminus\T$. Assuming that $\n'\geq\vb{m}$, we can find some $\n\in D$ with $\n\leq \n'$ which shares either the $n_1$ or $n_2$ projection of $\n'$. Assume that it is the former; the other case is similar. We can write $\n'=(x_1,0)\alpha_1+(0,x_2)(\alpha_2+\beta)+\offset$ and $\n=(x_1,0)\alpha_1+(0,x_2)\alpha_2+\y$ for some $\alpha_1,\alpha_2,\beta\in\N$. Notice that $\vp(\n')=\vp(\n)+\beta\la \a,(0,x_2)\ra$ and since $\vp(\n)\in\N$ it remains only to show that $\beta\la \a,(0,x_2)\ra\in\N$. Since $D$ is a wedge domain, the threshold sets which define its upper and lower boundary are not parallel, meaning that we can find two points in $D$, $\n_1$ and $\n_2$ such that $\n_1-\n_2=(0,x_2)\beta$ (indeed there are infinitely many such points). Thus, $\beta\la \a,(0,x_2)\ra=\vp(\n_1)-\vp(\n_2)\in\N$, completing the proof. 
\end{proof}

\section{Proof of Lemma \ref{lem:impossibility2}}
\label{app:impossibility2}
\begin{proof}
The proof is similar to that of Lemma \ref{lem:impossibility1}; we highlight only the differences. Here, choose the strictly increasing sequence $\{n_i\}$ such that 
$\vb{p}_i + (0,n_{i}) \in \dom_j$ and
$\vb{p}_{i+1} + (0,n_{i}) \in L$.
If $C$ is correct, then on input $\vb{p}_i + (0,n_{i})$, $\C$ can
first reach configuration $\c_i$ which has $\phi_i(\vb{p}_i)$ copies of $Y$ 
and then outputs $\vp_j(\vb{p}_i + (0,n_{i})) - \vp_i(\vb{p}_i)$ additional $Y$s. On input $\vb{p}_{i+1} + (0,n_{i})$, $\C$ can output $\phi_i(p_{i+1}) + \phi_i(p_i + (0,n_{i2})) - \phi_i(p_i) = \phi_i(p_{i+1}+(0,n_{i2}))
$ copies of  $Y$. The number of $Y$s produced is
greater than $\vp_L(\vb{p}_{i+1}+(0,n_i))$, since $\vp_L < \vp_i$ on
line $L$.  Thus $\C$ cannot be output-monotonic.
\end{proof}

\section{Proof of Lemma \ref{lem:line_offset}}
\label{app:technical1proof}
\begin{proof}
Recall that $\vp_i(\n)=\la \a_i,\n\ra+a_{0,i}$. Consider the line $J$ defined by all $\n$ such that $\la \a_i-\a_j,\n\ra=0$. Notice that the set of lines $\{\la \a_i-\a_j,\n\ra+\kappa:\kappa\in\Q\}$ are parallel. Therefore, if $L$ is parallel to $J$ then the hypothesis holds, so suppose not. In this case, for large enough $\n$, either $\la \a_i-\a_j,\n\ra>0$ or $\la \a_i-\a_j,\n\ra<0$ for all $\n\in L$. This implies that 
\begin{equation}
    \label{eq:lem_equal_coefficients1}
    \la \a_i-\a_j,\n\ra \xrightarrow[\n\text{ along }L]{\n\to\boldsymbol{\infty}}\infty,\quad\text{or}\quad \la \a_j-\a_i,\n\ra \xrightarrow[\n\text{ along }L]{\n\to\boldsymbol{\infty}}\infty.
\end{equation}

Now, we consider two cases based on the slope of $L$. First, suppose that $L$ has non-zero slope, i.e., is not a horizontal line. Let $\c(\n)$ be as in the definition of constant distance. We may assume that either  $\c(\n)\geq (0,0)$ for all $\n$ or  $\c(\n)\leq (0,0)$. Assume for now it is the former. 
We claim that we can find another bounded sequence, $\{\c(\n)'\}_{n\in L}$, such that for all $\n\in L$, $\n+\c(\n)'\in\dom_i$ and $\c(\n)'\leq 0$. Consider sliding each point of $\n\in L$ to a point further up on $L$, say $\n'$. Let $\c(\n)'$ satisfy $\n'+\c(\n)'=\n+\c(\n)\in\dom_i$. If $n'_1\geq n_1+c(\n)_1$ and $n'_2\geq n_2+c(\n)_2$, then we must have $\c(\n)'\leq (0,0)$. Moreover, if for all $\n$, $\n'$ is a constant distance away, then $\c(\n)'$ is bounded (by some function of $K$). Write $(K',K')\leq \c(\n)'\leq (0,0)$. 

Since $f$ is non-decreasing, we have that $\vp_j(\n)\leq \vp_i(\n+\c(\n))$ and $\vp_i(\n'+\c(\n'))\leq \vp_j(\n)$. 
The former implies that $\la \a_j-\a_i,\n\ra \leq a_{0,i}-a_{0,j}+\beta K$ for some constant $\beta$, and the latter that $\la \a_i-\a_j,\n\ra \leq a_{0,j}-a_{0,i}+\beta' K'$. However, this contradicts \eqref{eq:lem_equal_coefficients1}. This demonstrates that $L$ is parallel to $J$, so $\la \a_i-\a_j,\n\ra =\kappa$ for some $\kappa$. Moreover, $\kappa\in\Q$ since $\a_i,\a_j\in\Q^2$.  

If, on the other hand $L$ is a horizontal line, this implies that the lower boundary of both $\dom_i$ and $\dom_j$ are themselves lines (if they had positive slope, $L$ would not be a constant distance from either). Moreover, both $\dom_i$ and $\dom_j$ are two-way-infinite, thus they cannot be bounded above by any horizontal line. This implies that infinitely many points on $\dom_j$ are above the lower boundary of $\dom_i$ and vice versa. Hence, we can find two lines non-parallel lines with non-zero slope, $L_1$ and $L_2$. Applying what was proved above, we see that $\la \a_i,\n\ra =\la \a_j,\n\ra +\kappa_1$ for $\n\in L_1$ and $\la \a_i,\n\ra =\la \a_j,\n\ra +\kappa_2$ for $\n\in L_2$. However, if $L_1$ and $L_2$ are not parallel, this is possible iff $\a_i=\a_j$. 
\end{proof}

\section{Proof of Lemmas \ref{lem:equal_coefficients} and \ref{lem:equal_functions}}
\label{app:technical2proof}
\begin{proof}
We begin with Lemma \ref{lem:equal_coefficients}. Suppose that $\a_i\neq\a_j$. Consider the line $J$ consisting of those $\n$ such that $\la \a_i-\a_j,\n\ra=0$. $J$ is not parallel to either $I$ or $L$. Without loss of generality assume it is $L$. Thus, for large enough $\n\in L$, either $\la \a_i-\a_j,\n\ra>0$ or $\la \a_i-\a_j,\n\ra<0$. As in the proof of Lemma \ref{lem:line_offset}, taking limits as $\n$ gets large contradicts the existence of constant $\kappa$ such that $\la \a_i-\a_j,\n\ra=\kappa$. However, such a value is given by Lemma \ref{lem:line_offset}, which is thus a contradiction. 

As for Lemma \ref{lem:equal_functions}, we may consider an alternate representation of $f$ in which $\vp_i$ and $\vp_j$ are defined on $D$. The first part of the proof then dictates that $\a_i=\a_j$. Combined with the fact that $\vp_i(\n)=\vp_j(\n)$ for $\n\in D$ implies that $\vp_i=\vp_j$ everywhere. 
\end{proof}

\section{Proof of Claim \ref{claim:nonnegative-deficit}}
\label{app:proof_nonnegative-deficit}
\begin{proof}
Let $\z_p$ be the number of $(Z_1,Z_2)$ already consumed at the last
time that the deficit is cleared before $Z$-consumption stalls,
plus the number of $Z_1$ or $Z_2$ that are reactants of this last deficit clearing reaction, if any (i.e., $r_1$ if the
reaction is (\ref{eqn:y3}) or $r_2$ if the reaction is  (\ref{eqn:y4})).
Let $\n_p=(n_{1p},n_{2p})$ be the counts of the inputs $(X_1,X_2)$ that have been consumed at this time. 
The deficit is nonnegative when $Z$-consumption stalls if 
$f(M^{-1}({\bf z}_s)) \ge f(M^{-}({\bf z}_p))$. We will show that  $f(\n_p) \ge f(M^{-1}({\bf z}_p))$. Since 
$f(M^{-1}({\bf z}_s)) = f(\n)$ by Claim \ref{claim:equal}, since $\n \ge \n_p$ and since $f$ is increasing on integer-valued points, we then have
\[
f(M^{-1}({\bf z}_s)) = f(\n) \ge f(\n_p) \ge f(M^{-}({\bf z}_p)).
\]

To show that $f(\n_p) \ge f(M^{-1}({\bf z}_p))$, 
first suppose that ${\bf n}_p$ is on a fissure line.  Note that ${\bf n}_p$ and $M^{-1}({\bf z}_p)$ must be on the same fissure line, since deficit-clearing reactions can happen only in this case.  Now suppose that the deficit-clearing reaction applied is (\ref{eqn:y2}), i.e., ${\bf x} = \n_p$, ${\bf z} =M^{-1}({\bf z}_p)$ and both $\bx$ and $\z$ are on the same fissure line. So
\[
f(\n_p) = \min(\vp_A({\n_p}),\vp_B({\n_p})) -d_{\lx} \ge \min\{z_{1p},z_{2p} \} - d_{\lx} = f(M^{-1}({\bf z}_p)).
\]
Next suppose that the deficit-clearing reaction applied is (\ref{eqn:y3}).
Now, the line $l$ containing $M^{-1}(\z_p)$ must be such that
$l \ge K$; otherwise the condition that $l = \lx \mod K$ would not
hold.
Then by our choice of $K$, which is at least $k + d_{\max}$, it must be that
\[
(\vp_A(\n_p),\vp_B(\n_p)) \ge (z_{1p},z_{2p}+d_{\max}).
\]
Intuitively, this is because to "get back" to $\n_p$ from $M^{-1}(\z_p)$
requires consuming at least $d_{\max}$ more $Z_2$s.
Also $z_{1p}-z_{2p} = l > K$, so
\[
\vp_A(\n_p) \ge z_{1p} \ge z_{2p}+K \ge \min\{z_{1p},z_{2p}\}+(k+d_{\max})
\]
and so
\[
\min\{\vp_A(\n_p),\vp_B(\n_p)\} \ge \min\{z_{1p},z_{2p}\}+d_{\max}.
\]
Therefore
\[
f(\n_p) = \min\{\vp_A(\n_p),\vp_B(\n_p)\} - d_l
\ge \min\{z_{1p},z_{2p}\}+d_{\max} - d_l \ge \min\{z_{1p},z_{2p}\}
= f(M^{-1}(\z_p).
\]

Otherwise $\n_p$ is not on a fissure line (although $\lx$ might be the index of a fissure line). In this case $f(\n_p) = \min(\vp_A(\n_p),\vp_B(\n_p))$.
Then, regardless of which deficit-clearing reaction is applied, we have that
\[
f(\n_p) = \min(\vp_A(\n_p),\vp_B(\n_p)) \ge \min\{z_{1p},z_{2p} \} \ge f(M^{-1}({\bf z}_p)).
\]
Thus in every case we have that $f(\n_p) \ge f(M^{-1}(\z_p))$, and we are done.
\end{proof}

\section{Proof of Lemma \ref{lem:stiching}}
\label{app:proof_stiching}
{\bf Producing the $\nj{j}$:}
For simplicity of notation, fix some $j$, let $\dom' = \dom_j$,
let the base vectors of $\dom'$ be $(a_1',0)$ and $(0,a_2')$, and let
$\offset' = (o_1',o_2')$ be its offset.
Let $\nj{j} = \n' = (n_1',n_2')$ be the smallest
element of $\dom_j$ such that $\n \le \n'$.
The following CRN $C'$, which has nine reactions in total,
produces $n'_1$ copies of species $X_1'$ and $n'_2$ copies of
species $X_2'$ from a leader $L'$:
\[
\begin{array}{lll}
L' & \rightarrow L'_{00} + o_1 X_1' + o_2 X_2' & \\[.1in]
L'_{0b} + (o_1'+1) X_1 & \rightarrow L'_{1b} + a_1' X_1', & \mbox{ for $b=0$ and $b=1$} \\
L'_{b0} + (o_2'+1) X_2 & \rightarrow L'_{b1} + a_2' X_2', & \mbox{ for $b=0$ and $b=1$} \\
L'_{1b} + a_1' X_1 & \rightarrow L'_{1b} + a_1' X_1', & \mbox{ for $b=0$ and $b=1$} \\
L'_{b1} + a_2' X_2 & \rightarrow L'_{b1} + a_2' X_2', & \mbox{ for $b=0$ and $b=1$}
\end{array}
\]
To produce all inputs $\nj{1},\ldots, \nj{m}$, $m$ copies of CRN $C'$
are needed, each with independent copies of the species;
for example, species $X^{(j)}_1$ and $X^{(j)}_2$
are substituted for $X_1'$ and $X_2'$.
Three additional reactions produce
the leaders and input copies needed for each of these $m$ independent
copies of $C'$:
\[
\begin{array}{ll}
L & \rightarrow L'_1 + \ldots L'_m \\
X_1 & \rightarrow X^{(1)}_1 + \ldots X^{(m)}_1 \\
X_2 & \rightarrow X^{(1)}_2 + \ldots X^{(m)}_2.
\end{array}
\]


\noindent
{\bf Producing the output $y$:}
Let $Y_j$ be the output species of the output-oblivious CRN for $f_j$,
so that its count $y_j = f_j(\nj{j})$.
To produce $y = \min\{y_1, \ldots, y_m\}$ we need just one reaction:
\[
Y_1 + \ldots + Y_m \rightarrow Y.
\]
This completes the contruction for the case that all $\dom_j$ are 2D grids.

We now modify CRN $C$ with additional reactions to handle the case
that some of the $\dom_j$ may be 0D grids (i.e., points) or 1D grids.
The problem that can arise is that for some $j$,
$\dom' = \dom_j$ may not have any
point that is greater than input $\n$.
We add reactions
that generate a species $S'$ when this is the case.
Specifically, if $\dom'$ contains no base vector of the form $(a_1',0)$,
we add
\[
L'_{1b} \rightarrow S',  \mbox{ for $b=0$ and $b=1$},
\]
and we remove reactions above that involve $a_1'$.
Similarly, if $\dom'$ contains no base vector of the form $(0,a_2')$,
we add
\[
L'_{b1} \rightarrow S',  \mbox{ for $b=0$ and $b=1$}.
\]
and we remove reactions above that involve $a_2'$.
Note that $S'$ is produced if and only if
$\dom'$ has no point $\n'$ such that $\n \le \n'$.

As before, $C$ contains distinct copies of these reactions for each
$\dom_j$ that is not 2D, producing one or two copies of species $S_j$ if and
only if $\dom_j$ has no point $\nj{j}$ such that $\n \le \nj{j}$.

Finally, we need to ensure that the ouput $y$ produced is
the min of the $y_j$ taken over domains $\dom_j$ for which
$\nj{j} \in \dom_j$. To do this, we replace the single
reaction $Y_1 + \ldots + Y_m \rightarrow Y$ by
$2^m$ new reactions, one for each subset $I$ of $\{1,\ldots,m\}$.
The reaction corresponding to subset $I$ is of the form
\[
Y_1' + \ldots + Y_m' \rightarrow Y + Y_1'' + \ldots + Y_m'',
\]
where $Y_j' = Y_j$ if $j \in I$ and $Y_j' = S_j$ otherwise, and $Y_j'' $ is some inactive species if $j\in I$ and $Y_j'' = S_j$ otherwise.

\section{Non-increasing functions are not output-oblivious}
\label{app:non_increasing_functions}
Let $\C$ be a CRN stably computing $f$, and let $\outp=\{Y\}$. 
If $f:\N^2\to\N$ is non-increasing, then there exists $\n_1\leq\n_2$ such that $f(\n_1)>f(\n_2)$. Consider running $\C$ with $\n_2$ copies of the input. There exists a fair execution sequence in which $\n_1$ copies are first consumed by $\C$ which produces $f(\n_1)$ output species. After the remaining $\n_2-\n_1$ species are consumed, $\C$ will have to re-consume $f(\n_1)-f(\n_2)>0$ copies of $Y$. Hence $\C$ is not output-monotonic.

\end{document}